%% file: prober-arxiv.tex
\documentclass{TR-vldb}

\pdfoutput=1

\usepackage{times}
\usepackage{latexsym}
\usepackage{booktabs}
\usepackage{amsmath,amssymb}
\usepackage{defpaper2e}
\usepackage[algo2e,boxed]{algorithm2e}
\usepackage[boxed]{algorithm}
\usepackage{algorithmic}
\usepackage{graphicx}
\usepackage{multirow}
\usepackage[tight]{subfigure}

\renewcommand{\paragraph}[1]{\vspace{1mm}\noindent{\bf #1}}
\newtheorem{problem}{Problem}[section]

\newcommand{\miset}{MISet}
\newcommand{\maset}{MASet}

\newcommand{\seg}{{\bfseries{\tt sg}}}
\newcommand{\add}{{\bfseries{\tt ad}}}
\newcommand{\phn}{{\bfseries{\tt pn}}}
\newcommand{\jn}{{\bfseries{\tt jn}}}
\newcommand{\dup}{{\bfseries{\tt dp}}}
\newcommand{\sr}{{\bfseries{\tt sc}}}
\newcommand{\nm}{{\bfseries{\tt nm}}}
\newcommand{\wb}{{\bfseries{\tt wb}}}

\newcommand{\grush}{{\bfseries{\tt Bussiness}}}
\newcommand{\wpr}{{\bfseries{\tt Iterative}}}

\newcommand{\rbox}{\hfill $\Box$}

\newcommand{\bussi}{{\tt footwear}}
\newcommand{\actors}{{\tt actors}}
\newcommand{\books}{{\tt books}}

\newcommand{\senparty}{{\tt sen-party}}
\newcommand{\senstate}{{\tt sen-state}}
\newcommand{\mayor}{{\tt mayor}}

\newcommand{\alldocs}{{\tt All-recs}}
\newcommand{\wordor}{{\tt Wrd-OR}}
\newcommand{\wordand}{{\tt Wrd-AND}}

\newcommand{\op}{O}

\def\Rec{R}

\newcommand{\manyone}{{\tt many-to-one}}
\newcommand{\onemany}{{\tt one-to-many}}
\newcommand{\oneone}{{\tt one-to-one}}
\newcommand{\monotonic}{{\tt monotonic}}
\newcommand{\arbitrary}{{\tt arbitrary}}
\newcommand{\manymany}{{\tt many-to-many}}

\newcommand{\Pall}{P_{\mathit all}}
\newcommand{\Pimp}{P_{\mathit imp}}
\newcommand{\Pany}{P_{\mathit any}}
\newcommand{\Punion}{P_{\mathit uni}}
\newcommand{\Pint}{P_{\mathit int}}
\newcommand{\prober}{PROBER}

\newcommand{\Mall}{M_{\mathit all}}

\def\P{{P}}

\newcommand{\squishlist}{
   \begin{list}{$\bullet$}
    {
      \setlength{\itemsep}{1pt}
      \setlength{\parsep}{2pt}
      \setlength{\topsep}{2pt}
      \setlength{\partopsep}{0pt} } }
\newcommand{\squishend}{
    \end{list}  }

\numberofauthors{3}

\title{PROBER: Ad-Hoc Debugging of Extraction and Integration Pipelines}
\author{\alignauthor
Anish Das Sarma\\
       \affaddr{Yahoo, CA, USA}\\
       \email{anishdas@yahoo{\tt-}inc.com}       
\alignauthor
Alpa Jain\\
       \affaddr{Yahoo, CA, USA}\\
       \email{alpa@yahoo-inc.com}
\alignauthor
Philip Bohannon\\
       \affaddr{Yahoo, CA, USA}\\
       \email{plb@yahoo-inc.com}
}

\begin{document}
\maketitle

\begin{abstract}
Complex information extraction (IE) pipelines assembled by plumbing together off-the-shelf operators, specially customized operators, and operators re-used from other text processing pipelines are becoming an integral component of most text processing frameworks. A critical task faced by the IE pipeline user is to run a post-mortem analysis on the output. Due to the diverse nature of extraction operators (often implemented by independent groups), it is time consuming and error-prone to describe operator semantics formally or operationally to a provenance system.  

We introduce the first system that helps IE users analyze pipeline semantics and infer provenance interactively while debugging.  This allows the effort to be proportional to the need, and to focus on the portions of the pipeline under the greatest suspicion. We present a generic debugger for running post-execution analysis of any IE pipeline consisting of arbitrary types of operators. We propose an effective provenance model for IE pipelines which captures a variety of operator types, ranging from those for which full or no specifications are available. We present a suite of algorithms to effectively build provenance and facilitate debugging. Finally, we present an extensive experimental study on large-scale real-world extractions from an index of $\sim$500 million Web documents.

\eat{
Complex information extraction (IE) pipelines inter-connecting a diverse set of operators (e.g., text segmentation, entity identification, record construction and validation) are becoming an integral component of most text processing frameworks. Typically, such IE pipelines are assembled by plumbing together off-the-shelf operators, specially customized operators, and operators re-used from other text processing pipelines. After processing a set of web pages to generate output records, a critical task faced by the IE pipeline user is to run a post-mortem analysis on the output: information extraction methods are inherently imprecise in nature and naturally a pipeline that combines one or more such methods also suffers from similar impreciseness. Therefore, post-execution analysis is an essential exercise carried out not only during the development of an IE pipeline but also post-development. Debugging an IE output is challenging for several reasons. First, we need to build an effective method to trace and link enough information regarding each output record. Second, operators in an IE pipeline may substantially differ in their behaviour and their impact on the output. Furthermore, exact specifications of each operator (e.g., off-the-shelf operators) may not available. Finally, due to limited human resources, an effective debugger must restrict the amount of input necessary to resolve an output record.

In this paper, we present a generic debugger for running post-execution analysis of any IE pipeline consisting of arbitrary types of operators. We propose an effective provenance model for IE pipelines which captures a variety of operator types, ranging from those for which full or no specifications are available. We present a suite of algorithms to effectively build provenance and facilitate debugging. Finally, we present an extensive experimental study on large-scale real-world extractions from an index of $\sim$500 million Web documents.
}
\end{abstract}

\input{1-introduction.tex}

\input{2-background.tex}
\input{4-provenance.tex}
\input{5-algorithms.tex}
\input{consolidation.tex}

\input{6-experiments.tex}
\input{7-conclusion.tex}

\bibliographystyle{abbrv}
{\scriptsize
\bibliography{Bibs/conferences,Bibs/names,Bibs/journals,Bibs/misc}
}

\input{appendix.tex}

\end{document}

%% file: 1-introduction.tex
\vspace{-2mm}
\section{Introduction}
\label{sec:intro}

Growing amounts of knowledge is being made available in the form of unstructured text documents such as, web pages, email, news articles, etc. Information extraction (IE) systems identify structured information (e.g., people names, relations betwen companies, people, locations, etc.) and, not surprisingly, IE systems are becoming a critical first-class operator in a large number of text-processing frameworks. As a concrete example, search engines are moving beyond a ``keyword in, document out'' paradigm to providing structured information relevant to users' queries (e.g., providing contact information for businesses when user queries involve business names). For this, search engines typically rely on having available large repositories of structured information generated from web pages or query logs using IE systems. With the increasing complexity of IE pipelines, a critical exercise for IE developers and even users is to {\em debug}, i.e., perform a thorough post-mortem analysis of the output generated by running an entire or partial extraction pipeline. Despite the popularity of IE pipelines, very little attention has been given to building effective ways to trace the control or data flow through an extraction pipeline.

\vspace{-2mm}
\begin{example}
\label{ex:one}
Consider an IE pipeline for extracting contact information for businesses, namely, business name, address (one or many), phone number (one or many), from a set of web pages. The pipeline, in addition to others, consists of operators (a) to clean and parse html web pages, (b) to classify `blocks' of text in a web page as being useful or not for this task, (c) extract business names, (d) extract address(es). (We discuss this real-world pipeline in detail later in Section~\ref{sec:background}.) Two interesting points to note here: First, in practice, such complex pipelines may be put together using off-the-shelf operators (e.g., html parsers or segmenters) along with some newly designed as well as some re-usable operators from other systems. Second, IE is an erroneous process and oftentimes, output from an IE pipeline may miss some information (e.g., a record where contact information is present but business name is absent) or may generate unexpected output (e.g., associate a fax number with a business instead of phone number). 

Say a user of this IE pipeline processes a batch of web pages and generates a set of (partial, complete, or incorrect) output records. Given the output, the user may be interested in understanding why certain incorrect records were generated to identify and eliminate their `sources'; similarly, the user may also be interested in understanding why certain records were missing  attributes in the output to identify the `restrictive' operators in the pipeline.
\end{example}
\vspace{-2mm}

To date, there have been two main approaches for understanding the output from an IE but neither fully addresses the problem of debugging arbitrary IE pipelines. The first approach is to build statistical models to predict the output quality of an IE system~\cite{icde08/jain,goodk/jain}. However, these models address the more modest goal of assessing the overall output quality and lack the intuitive interaction necessarily for building debuggers to trace the generation of an output record. The second approach involves using complete knowledge of how each operator functions. As highlighted by the above example, prior information regarding the specifications of the operators may not be available (e.g., off-the-shelf black-box operators). In the absence of full function specifications of an operator, the only (straightforward) approach to debugging is exploring all data in the pipeline. However,  such an approach is clearly infeasible due to the sheer volume of data. For instance, debugging a simple pipeline involving 10 operators with 10,000 input records per operator would require 100K records to be manually examined. (As we shall see in our experiments in Section~\ref{sec:experiments}, typical data sizes are even larger.)


This paper presents PROBER (for Provenance-Based Debugger),  the first generic framework for debugging information extraction pipelines composed of arbitrary (``black-box'') operators. \eat{Several challenges need to be overcome in order to provide effective ad-hoc debugging functionality for  IE pipelines: }A critical task towards building debuggers is that of tracing and linking output records from each operator and understanding their transformations across different operators in the pipeline. To trace the lineage of any arbitrary record in the output, we propose a novel provenance model for IE pipelines. With debugging in mind, our provenance model tries to minimize the amount of user effort necessary in resolving the fate of the records in the output. For example, provenance for (incorrect) output records {\em only} refer to input tuples that impacted this output record. We present a suite of algorithms to build the provenance for an IE run; our algorithms explore the tradeoff between efficiency of building provenance, and the amount of information captured by it. 

As outlined by Example~\ref{ex:one}, exact functional specifications for operators in an IE pipeline may not be available. However, provenance building can exploit various properties of these operators that can be learned by sampling, namely, monotonoic, one-to-one, one-to-many, or arbitrary. We characterize a diverse set of operators, and their properties, that are found in real-world extraction pipelines and rigorously examine methods to build provenance information for each of these combinations. 

 
In summary, beyond the conceptualization  of \prober\ (Section~\ref{sec:background}), the main contributions of the paper are as follows:
\vspace{-2mm}
\begin{itemize}\itemsep-0.03in
\item A novel provenance model for the task of debugging information extraction pipelines. Our model effectively accounts for scenarios where incomplete (or no) knowledge about the underlying operators in the pipeline is available (Section~\ref{sec:prov}).
\item A suite of effective algorithms to build provenance, given an IE run (Section~\ref{sec:algos}). 
\item An end-to-end solution combining extraction output along with provenance information for debugging (Section~\ref{sec:alltogether}).
\item An extensive evaluation over real-world datasets, demonstrating the effectiveness of our framework in debugging extraction pipelines (Section~\ref{sec:experiments}). 
\end{itemize}



%% file: 2-background.tex
\section{Problem Formulation}
\label{sec:background}
\label{sec:motivation}
While IE pipelines may vary in their implementation logic~\cite{dl00/agichtein,cikm09/kasneci,aaai06/pasca,acl06/pantel} several underlying common components can be abstracted from the implementation details. We characterize information extraction pipelines for the task of performing post-mortem analysis.


\vspace{-2mm}
\begin{nameddefinition}{Record}
A {\em record} $r$ is a basic unit of data (e.g., a tuple), consisting of a globally unique identifier $I(r)$, and value $V(r)$. We use $\Rec$ denote the set of all records. 
\end{nameddefinition}
%


\vspace{-5mm}
\begin{nameddefinition}{Operator}
An {\em operator} is defined by a function $O:(I_1, I_2, \cdots, I_N)\rightarrow R$, where each $I_i\subseteq \Rec$ is a set of records. In practice, the function $O$ may be unknown to us.
\end{nameddefinition}
%
\vspace{-1mm}
\noindent Intuitively, an operator takes as input an $N$-tuple of sets of records and outputs one set of records. 
Specifications on how an operator generates an output record may be available in varying forms. Specifically, we consider the following four scenarios involving operator specifications. 
%
%
\begin{figure}[t]
\centering
\vspace{-4mm}
\includegraphics[width=0.7\columnwidth]{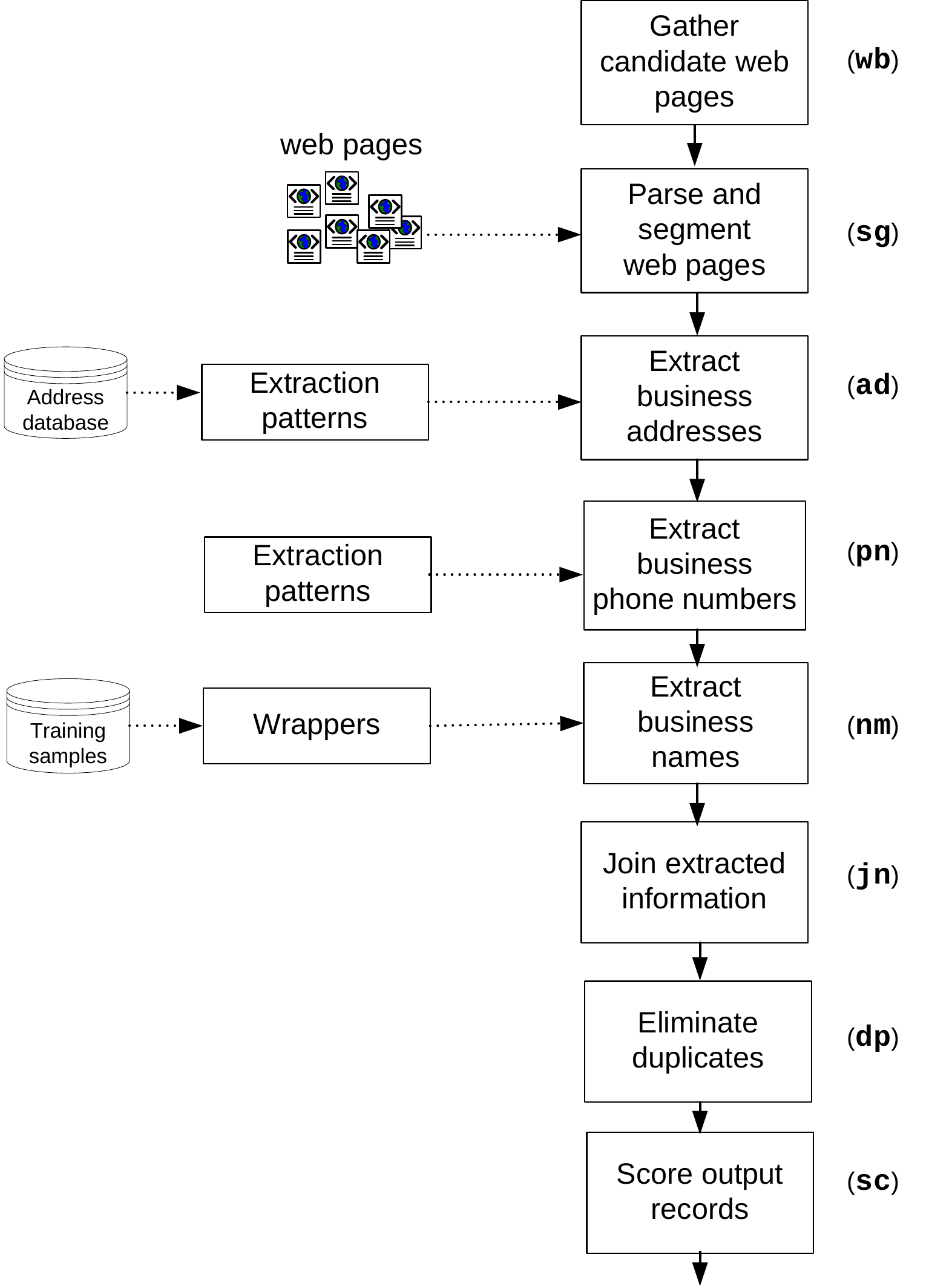}
\vspace{-5mm}
\caption{Example of an IE pipeline to generate business names and their contact information.}
\label{fig:goldrush}
\vspace{-5mm}
\label{fig:example}
\end{figure}

An operator is said to be a {\em black-box} if we have no information about it. In this case, naturally, the only way to gain information about a black-box operator is by executing it on input sets of records. In contrast, we have {\em exact} information about an operator $O$ if we know precisely {\em which} input records contributed to each output record, {\em and how}. We have {\em Input-Output (IO) specifications} when for each output record, we know which input records were used to construct it, however exactly how a record is generated is unknown to us. Finally, we may have {\em integrity constraints}, e.g., key-foreign key relationships, satisfied by the input and output records. For instance, an operator may support a `debug' mode where each output record is assigned an id associated with the input records that generated it. Effectively, using key-foreign keys we have the same information as that in IO specifications, but this information is (indirectly) available using dependencies on the values of {\em fields} in input and output records. 
Next, we define various (standard) properties of an operator, that help design specialized algorithms for building provenance and debugging effectively. As we will see later, these properties may be learned by sampling or the operator specifications (when available) described above. 

\vspace{-2mm}
\begin{nameddefinition}{Properties}
\vspace{-1mm}
\label{def:props}
\begin{itemize}
\itemsep-2pt
{\footnotesize
\item {\bf \monotonic:} Operator $O$ is \monotonic\ iff $\forall I_1,I_2\subset \Rec: (I_1\subseteq I_2)\Rightarrow (O(I_1)\subseteq O(I_2))$.
\item {\bf \oneone:} Operator $O$ is \oneone\ iff: (a) $\forall I\subset\Rec: O(I) = \bigcup_{r\in I} O(\{r\})$; (b) $\forall r\in \Rec$: $|O(\{r\})|\leq 1$.
\item {\bf \onemany:} Operator $O$ is \onemany\ iff $\forall I\subset\Rec: O(I) = \bigcup_{r\in I} O(\{r\})$.
\item {\bf \manyone:} Operator $O$ is \manyone\ iff $\forall I\subset\Rec$, $\exists$ a partition $P_I=\{I_1, \ldots, I_n\}$ 
of $I$\footnote{(a) $I = \bigcup_{i=1}^{n} I_i$ (b) $\forall i\neq j: (I_i\cap I_j)=\emptyset$.} 
such that: (a) $O(I)=\bigcup_{i=1}^{n} O(I_i)$; and (b) $\forall i: |O(I_i)|\leq 1$.
}\end{itemize}
\vspace{-4mm} \end{nameddefinition}
\vspace{-6mm}


\begin{nameddefinition}{Extraction Pipeline}
An {\em extraction pipeline} $P$ is defined by a DAG $G(V,E)$ consisting of a set $V$ of nodes and a set $E$ of edges where each node $v\in V$ corresponds to an operator $O$ in the pipeline. An edge $a\rightarrow b$ between nodes $a$ and $b$ indicates that the output from the operator represented by $a$ is input to operator represented by $b$. We have a single special node $s\in V$ with no incoming edges representing the operator that takes input to the pipeline, and one special node $t\in V$ with no outgoing edges representing the operator that outputs the final set of records.
\end{nameddefinition}
We now discuss a real-world extraction pipeline (used at Yahoo!), which forms the basis of our illustrations in this paper. 
\subsection*{Motivating Example}

Figure~\ref{fig:goldrush} shows a real-world extraction pipeline, \grush, for building a large collection of businesses (see Example~\ref{ex:one}) by extracting records of the form $\langle n, a, p\rangle$, where business $n$ is located at address $a$ with contact number $p$. The first step is to build a set of web pages likely to contain information regarding businesses which is done using a variety of document retrieval strategies. Specifically, we issue manually generated domain-dependent queries (e.g., ``Toyota car dealership locations'') as well as use form filling methods where entries such as, model, make, and zipcode, may be filled in order to fetch a list of car dealerships. This operator, denoted by \wb\ is an example of a black-box operator with arbitrary properties. 

Given a collection of web pages, operator \seg\ parses the html page and identifies appropriate segments of text in this page, where ideally, each segment contains a complete target record (see Figure~\ref{fig:sample-data} in the appendix for a real-world example). These segments are then processed by operators, \add\ and \phn, which respectively identify an occurrence of an address and a phone number. The annotation from one operator is used by the subsequent operator to identify regions of text that should not be processed. \add\ and \phn\ are implemented using hand-crafted patterns based on a dictionary of address formats. The \nm\ operator on the other hand needs to identify names of business which may be arbitrary strings and for this, we follow a wrapper-induction appraoch. In particular, using some training examples we learn a wrapper rule to identify candidate business names; these rules are based on the document structure of the html content. Of course, several other implementations for each of these operators are possible and the implementation details are orthogonal to our discussion since our goal is to build debuggers for pipelines with black-box operators where no implementation information may be available. The \jn\ operator joins outupt from \add, \phn, and \nm\ to build candidate output records which are, in turn, processed by \dup\ to eliminate duplicates. The final operator, assignes a confidence score \sr\ to each output record. 

\sloppy We note that all our implementations of the above operators are monotonic. (Obviously, there may be non-monotonic implementations in other pipelines, but we primarily consider monotonic operators in this paper.) Although monotonic, the operators from the pipeline span a variety of properties, e.g., segmentation is a \onemany\ operation, and by design one address is extracted from each segment, so address extraction is \oneone, while de-duplication is \manyone. Candidate webpage generation and wrapper training, on the other hand are arbitrary, i.e. ``\manymany''.



Given unexpected output records, an IE developer may want to answer some natural questions about the output. (Figure~\ref{fig:sample-data} in the appendix shows an example where \seg\ generates an incorrect segment that leads to missing one address and extracting one incorrect address.) Specifically, a developer may be interested in tracing all or part of the input records that contributed to a particular output record. For instance, given an incorrectly extracted record, we would like to know only the relevant subset of webpages and training data that impacted it, i.e., the {\em minimal} amount of input data necessary to identify the error. Motivated by the above observations, we focus on the following problem in this paper. 

\vspace{-2mm}
\begin{problem}
Given a pipeline $P$, input $I$, and partial information about operators in $P$, we would like to (1) build {\em provenance} for the set of (intermediate and final) records in the pipeline; (2) expose provenance to developers through a query language and guide them in debugging the pipeline.
\end{problem}


%% file: 4-provenance.tex
\section{Provenance for IE Pipelines}
\label{sec:prov}

The notion of provenance is relatively well-understood for traditional relational databases (refer~\cite{ikeda-survey,provenance}). A commonly advocated model~\cite{uldb06} is to use a boolean-formula provenance, e.g., $S_1\wedge (S_2\vee \neg S_3)$. For the purpose of debugging extraction pipelines, such provenance models are not appropriate for two main reasons. First, unlike relational queries where the exact specifications of each operator are known, we may have black-box operators in our extraction pipeline. Second, for debugging, ideally we would like to limit the number of records (and simplify their interdependencies typically represented as boolean formulas for relational operators) a human has to assess in order to understand the issue at hand. With these in mind, we define a provenance model based on {\em minimal subsets} of operator inputs that capture necessary information (Section~\ref{subsec:upm}) and extend this basic model to operators where multiple minimal subsets may exist (Section~\ref{subsec:cp}) .

\subsection{MISet: Basic Unit of Provenance}
\label{subsec:upm}

To define the provenance of an IE pipeline $P$, we begin by defining the provenance for each operator in $P$;
the subsequent sections show how to construct the provenance for each operator in $P$ (Section~\ref{sec:algos}) and for $P$ by composing individual operators' provenance (Section~\ref{sec:alltogether}). We primarily confine ourselves to extraction pipelines consisting of only monotonic operators (see Definition~\ref{def:props}), which are a common case in practice (as in our motivating example from Section~\ref{sec:motivation}); extensions to non-monotonic operators is very briefly discussed in the appendix (Section~\ref{app:extensions}), but largely left as future work. 

We define the provenance of an extraction operator $O$ based on the provenance for each output record $r\in R$ for $O$. Ideally, we would like the provenance of $r$ to represent precisely the set of records that contributed to $r$, however, as we will see, in practice it may not be possible to always determine the precise set of contributing records (e.g., in the absence of exact information about $O$), and even if possible it may be computationally intractable. For our goal of building a debugger, we observe that one of the main operations we expect users to perform is look at an (erroneous) output record $r$, and explore its provenance to determine the cause of this error. Therefore, a suitable provenance model is one that enables users to examine the fewest records required to decide the fate of an output record $r$. Formally, we define a basic unit of provenance, called {\miset} as follows:
\vspace{-1mm}
\begin{nameddefinition}{MISet}
Given an operator $\op$, its input $I$ and output $R$, we say that $I_s\subseteq I$ is a {\em Minimal Subset} (\miset) of $r\in R$ if and only if: (1)  $r\in \op(I_s)$; and (2) $\forall I'\subset I_s: r\not\in \op(I')$. We use $\Mall(\op,I,r\in R)$ to denote the set of all \miset s of $\op$ for input $I$ and output record $r\in R$.
\end{nameddefinition}
\vspace{-1mm}
\noindent Intuitively, an \miset\ gives the fewest input records required for a particular output record $r$ to be present. Therefore, an \miset\ provides users with one possible reason for the occurrence of $r$. This, in turn, reduces the burden of manual annotation on the users; in the absence of \miset s, a user may have to explore the entire input to understand what caused an error in the output. The notion of \miset s primarily focuses on debugging the {\em presence} of records in the output;  in Appendix~\ref{app:extensions} we briefly discuss a corresponding notion (\maset s) for the case of non-existence of records in the output, but leave details for future work. 


In practice, we may have more than one \miset\ possible for an output record as shown by the following example.
\vspace{-1mm}
\begin{example}
\label{ex:miset}
Consider a (simple) record validation operator (e.g., \sr\ in Figure~\ref{fig:example}) that computes the ``support'' of each record and outputs only records with sufficient support. Suppose \sr\ outputs a record $r$ if there are atleast 50 input records supporting it. Given an input of 100 records that could support $r$, the \miset\ for $r$ is {\em any} subset of the input records of exactly 50 records.
\end{example}
\vspace{-1mm}
\noindent When multiple \miset s are available, several ways of building provenance are possible, each differing in the extent to which they impart infomation and execution speed, as explored next. 

\subsection{Handling Multiple MISets}
\label{subsec:cp}

Several formalisms for provenance model are possible when multiple \miset s are available. We rigorously examine compositions of \miset s, while capturing the spectrum of complete (and potentially intractable) provenance, to more tractable (but approximate) provenance. Later in Section~\ref{subsec:algos}, we will present algorithms for building each of type of provenance. 

Consider an operator $\op$ which consumes input $I$ and generates output $R$; for a record $r \in R$, we denote the provenance of $r$ as $\P(\op, I, r)$. We use subscripts $\P_{*}$ to capture various types of provenance and when clear from the context, we simply use $\P_{*}(r)$ to denote $\P_{*}(\op, I, r)$.

\paragraph{\bf All- and Any-provenance:} Ideally for any output record, we would like to provide all possible information using \miset s, i.e., capture all possible ``causes'' of an output record. 
\vspace{-1mm}
\begin{nameddefinition}{All-provenance}\label{def:pall}
Given an operator $\op$, input $I$, output $R$, and $r\in R$, we define {\em all-provenance} as $\Pall(r) = \Mall(\op,I,r\in R)$.
\end{nameddefinition}
\vspace{-1mm}

\noindent In many cases $\Pall$ may be intractable to compute or store, and we may need to resort to ``approximations'' of it, presented shortly. Alternatively, we may want to find any one (or $k$) \miset s.

\vspace{-1mm}
\begin{nameddefinition}{Any-provenance}\label{def:pany}
Given integer $k>0$, operator $\op$, input $I$, output $R$, and $r\in R$, we define {\em any-provenance} as any $\Pany(r) \subseteq \Pall(r)$ of size $\min(k,|\Pall(r)|)$.
\end{nameddefinition}
\vspace{-1mm}

\paragraph{\bf Impact-provenance:} Given restricted amount of editorial resources, we may want to explore the most {\em impactful}, i.e., top-$l$ input records sorted by their impact instead of any-$k$ \miset s. Our next definition of provenance ranks tuples based on their expected impact on the output record, measured by the number of \miset s in which a tuple is present.
\vspace{-1mm}
\begin{nameddefinition}{Impact-provenance}\label{def:pimp}
Given an operator $\op$, input $I$, output $R$, and $r\in R$, we define {\em impact-provenance} as $\Pimp(r) = \{(i,c_i) | \exists M\in \Pall, i\in M, c_i=\sum_{M\in \Pall, i\in M}1 \}$.
\end{nameddefinition}
\vspace{-1mm}

\paragraph{\bf Union- and Intersection-provenance:} Our next goal is to summarize $\Pall$ using two approximations: (1) We obtain an ``upper bound'' provenance that captures the set of all possible inputs {\em possible} for $r$, instead of exact combinations of inputs. Therefore, we define the union-provenance of $r$ to be the union of all its \miset s.  (2) We obtain a ``lower bound'' provenance that captures the set of all possible inputs {\em necessary} for $r$; we define the intersection-provenance of $r$ to be the common input records among all \miset s.
\vspace{-1mm}
\begin{nameddefinition}{Union-Provenance}\label{def:punion}
Given an operator $\op$, input $I$, output $R$, and $r\in R$, we define {\em union-provenance} as $\Punion(o) = \bigcup_{I_s\in \Mall(\op,I,r\in R)} I_s$.
\end{nameddefinition}
\vspace{-1mm}
\vspace{-1mm}
\begin{nameddefinition}{Intersection-Provenance}\label{def:pint}
Given an operator $\op$, input $I$, output $R$, and $r\in R$, we define {\em intersection-provenance} as $\Pint(o) = \bigcap_{I_s\in \Mall(\op,I,r\in R)} I_s$.
\end{nameddefinition}
\vspace{-1mm}

\noindent It can be seen easily that for operators with unique \miset s, $\Punion$ and $\Pint$ coincide.

\eat{
\begin{lemma}
For a \manyone\  (and hence a \oneone) operator $\op$, we have $\forall I\subset \Rec, R=\op(I):  \Punion(\op,I,r\in R) =\Pint(\op,I,r\in R)$.
\end{lemma}
\begin{proof}
easy. Extend to other properties?
\end{proof}
}

%% file: 5-algorithms.tex
\section{Inferring Provenance}
\label{subsec:provconstruction}
\label{subsec:propsummary}
\label{subsec:algos}
\label{sec:algos}
\input{results-algos-complexity.tex}

We now turn to the critical task of deriving the provenance formalisms proposed in Section~\ref{subsec:cp}. We primarily consider the generic black-box operator case while rigorously examining various properties described in Section~\ref{sec:background}, and our carry over for more specific cases (e.g., exact or I/O specfications); however whenever necessary we will point out the differences. Table~\ref{tab:summary} summarizes the results on provenance inference achieved in this section, with details in the following subsections. All time complexities in the table are given in terms of the size of the input $N$ and the size of the output $M$. The table gives time complexity assuming an operator can be executed freely (i.e., in ${\cal O}(1)$); depending on the running time of the operator, the appropriate factor can be multiplied.


\subsection{Unique \miset\ Operator} 
We begin with the case of when a combination of input $I$, output $R$, and operator $O$ has a unique \miset\ for each output record. Note that we don't assume that we know the uniqueness of \miset s; instead, we only need the {\em existence} of a unique \miset. (That is, our results hold even in the case when we don't have any information about the black-box operator, but it just happens that the operator functions in a way that creates a unique \miset\ for output records.) We have the following main result for unique \miset\ operators. (Complete proofs for all results in the paper are presented in Appendix~\ref{sec:proofs}.)
\begin{theorem}[Compute \miset]\label{thm:umset}
Given any monotonic operator $\op$, input $I$, and output $R$, if $\op$ has a unique \miset\ for each output record $r\in R$, then the unique $\miset$ for $r$ can be constructed in ${\cal O}(N)$.
\end{theorem}
\noindent As a consequence, all entries in Table~\ref{tab:summary} for unique \miset\ operators can be solved in ${\cal O}(N)$. The above result follows from: (a) the following lemma that tests for uniqueness of \miset s; and (b) the fact that a single \miset\ for any monotonic operator can be computed efficiently using Algorithm~\ref{algo:anyset} (Lemma~\ref{lem:anyset} presents a more general result for $k$ \miset s computation shortly).
\begin{lemma}[Uniqueness Test]\label{lem:utest}
Given any monotonic operator $\op$, input $I$, output $R$, and any $r\in R$, Algorithm~\ref{algo:unique} tests in ${\cal O}(N)$ whether there is a unique \miset\ for $r$.
\end{lemma}
\noindent The algorithm works in two stages. First, it finds any \miset\ $M$ for $r$. Second, it attempts to find other sets that might produce $r$ by applying $\op$ on the entire input except one record from $M$. If none of these sets produce $r$, there is no other \miset.
\begin{algorithm}[t]
{\scriptsize
\caption{\scriptsize Testing the uniqueness of \miset s for any monotonic operator.}
\begin{algorithmic}[1]\label{algo:unique}
\REQUIRE $O, I, r\in O(I)$
\STATE Find an \miset\ $M$ (Algorithm~\ref{algo:anyset})
\FOR{$m \in M$}
          \IF{$r\in O(I-\{m\})$}
              \STATE {RETURN ``Non-unique''}
          \ENDIF
\ENDFOR
\ RETURN ``Unique''
\end{algorithmic}
}
\end{algorithm}
\subsection{One-to-One and One-to-Many Operators}
For the relatively simple cases of \oneone\ and \onemany\ operators, we can obtain easily our composite provenances in linear time in the size of the input $N$ for \oneone\ operators and linear time in $N+M$ for \onemany\ operators, as shown by the following theorem.
\begin{theorem}\label{thm:111nall}
Given a \oneone\ operator $\op$, input $I$ of size $N$, and an output record $r\in R$, each of $\Pall(r)$, $\Pany(r)$, $\Punion(r)$, $\Pint(r)$, and $\Pimp(r)$ can be computed in ${\cal O}(N)$. For a \onemany\ operator, the complexity is increased to ${\cal O}(N+M)$, where $M=|O(I)|$.
\end{theorem}
\subsection{Many-to-one and Arbitrary Operators}
\begin{algorithm}[t]
{\scriptsize
\begin{algorithmic}[1]\label{algo:anyset}
\REQUIRE $O, I, r\in O(I)$
\STATE Set $M=I$
\FOR{$m\in M$}
          \IF{$r\in O(M-\{m\})$}
		\STATE M=M-\{m\}
	\ENDIF
\ENDFOR
\ RETURN M
\end{algorithmic}
\caption{\scriptsize Computing a single \miset\ for any monotonic operator.}
}
\end{algorithm}
\paragraph{\bf Computing $\Pany$:} We can always find an \miset\ using an ${\cal O}(N)$ algorithm for {\em any} monotonic operator $\op$. Algorithm~\ref{algo:anyset} describes how to find such an \miset. Algorithm~\ref{algo:kset}  provides an extension that finds $k>0$ \miset s (when $k$ \miset s exist): For brevity, we specify the algorithm to find a different $(p+1)$th \miset\ given $p$ \miset s. The algorithm is invoked $(k-1)$ times after Algorithm~\ref{algo:anyset} successively adding \miset s to obtain $k$ distinct \miset s. The following lemma establishes our result.
\vspace{-2mm}
\begin{lemma}\label{lem:anyset}
Given any monotonic operator $\op$, input $I$, and output record $r\in \op(I)$, $\Pany$ for $k$ \miset s can be computed in ${\cal O}(N^{k+1})$.
\end{lemma}
\vspace{-2mm}
\begin{algorithm}[t]
{\scriptsize
\caption{\scriptsize Algorithm for finding an \miset\ different from a given set of $p$ \miset s. The algorithm can be applied multiple times to generate several distinct \miset s.}
\begin{algorithmic}[1]\label{algo:kset}
\REQUIRE $O, I, r\in O(I)$, $p$ \miset s $\{M_1, \ldots, M_p\}$
\FOR{$(m_1,\ldots, m_p)\in M_1\times \ldots \times M_p$}
\STATE Set $I' = I-\{m_1,\ldots,m_p\}$
\IF{$r\in O(I')$}
\STATE RETURN Algorithm~\ref{algo:anyset} result using $O, I', r$ as input.
\ENDIF
\ENDFOR
\STATE RETURN ``No other \miset''
\end{algorithmic}
}
\end{algorithm}
\noindent Note that the actual complexity is ${\cal O}(N\alpha^{k})$ where $\alpha$ is the size of the largest \miset. Therefore, if all \miset s are small, the algorithm runs very efficiently.

\paragraph{\bf Computing $\Pint$:} $\Pint$ for {\em any} arbitrary operator can be computed using an ${\cal O}(N)$ algorithm. Algorithm~\ref{algo:pint} shows how to obtain $\Pint$ and the theorem below establishes the result.
\begin{algorithm}[t]
{\scriptsize
\caption{\scriptsize Computing $\Pint(r)$ for any monotonic operator $\op$.}
\begin{algorithmic}[1]\label{algo:pint}
\REQUIRE $O, I, r\in Op(I)$
\STATE Set $S=\emptyset$
\FOR{$i \in I$}
          \IF{$r\not\in O(I-\{i\})$}
              \STATE $S = S\cup \{i\}$
          \ENDIF
\ENDFOR
\ RETURN $S$.
\end{algorithmic}
}
\end{algorithm}
\vspace{-1mm}
\begin{theorem}[$\Pint$ Computation]
\label{thm:pint}
Given any monotonic operator $\op$, input $I$, and output record $r\in O(I)$, Algorithm~\ref{algo:pint} correctly computes the $\Pint(r)$ with ${\cal O}(M+N)$ executions of $\op$.
\end{theorem}
\vspace{-2mm}
\paragraph{\bf Computing $\Punion$:} To compute the $\Punion$ of an output record $r\in \op(I)$, for each input record $i\in I$, we need to determine whether there exists any \miset\ $M$ containing $i$. We employ a simple approach to determining if there exists any \miset\ with $i$. We simply find all \miset s (using Algorithms~\ref{algo:anyset} and~\ref{algo:kset}), and determine their union. Note that in the worst case, our naive algorithm takes exponential time in $|I|$. Finding better upper or matching lower bounds is an open problem.

\eat{
\begin{algorithm}[t]
\begin{small}
\caption{Algorithm for finding $\Punion$ for a \manyone\ monotonic operator.}
\begin{algorithmic}[1]\label{algo:n1union}
\REQUIRE $O, I, r\in O(I)$
\STATE Set $S=\emptyset$
\FOR{$i \in I$}
\STATE Find all \miset s $\{M_1, \ldots, M_k\}$ for $O, \{I-i\}
\STATE [[Need to write]]
\ENDFOR
\ RETURN $S$
\end{algorithmic}
\end{small}
\end{algorithm}

\begin{theorem}
Given any \manyone\ monotonic operator $\op$, input $I$, $|I|=N$, and an output $r\in R$, Algorithm~\ref{algo:n1union} finds $\Punion(r)$ in ${\cal O}(N2^N)$.
\end{theorem}
\begin{proof}
proof.
\end{proof}
}

\paragraph{\bf Computing $\Pall$:} \sloppy For $\Pall$, we show that finding $\Pall$ for an \manyone\ monotonic operator is \#P-complete\footnote{\#P-completeness corresponds to the class of hard counting problems.} in the size of the input; i.e., there does not exists any polynomial-time algorithm to compute $\Pall$ exactly. Our result is proved using a reduction from the problem of finding all {\em Minimal Vertex-Cover} in an undirected graph.

\begin{theorem}\label{thm:n1hard}
Given any \arbitrary\ or \manyone\ monotonic operator $\op$, input $I$ of size $N$, output $R$ of size $M$, and an output $r\in R$, it is \#P-complete in $N$ and $M$ to compute $\Pall$.
\end{theorem}
\noindent Shortly, we show that the complexity can be made PTIME by restricting ourselves to \miset s of bounded size.

\paragraph{\bf Computing $\Pimp$:} Finally, we show that the hardness result $\Pall$ can be extended easily to $\Pimp$.
\begin{corollary}\label{cor:n1hard}
Given any \arbitrary\ or \manyone\ monotonic operator $\op$, input $I$ of size $N$, output $R$ of size $M$, and an output $r\in R$, it is \#P-complete in $N$ and $M$ to compute $\Pimp$.
\end{corollary}


\eat{
\subsection{Arbitrary Monotonic Operator}
Our hardness result for $\Pall$ for \manyone\ operators carries over for arbitrary monotonic operators.
\vspace{-2mm}
\begin{theorem}\label{thm:arbhard}
Given any monotonic operator $\op$, input $I$ of size $N$, output $O$ of size $M$, and an output $o\in O$, it is \#P-complete in $N$ and $M$ to compute $\Pall$.
\end{theorem}
\vspace{-2mm}

\vspace{-2mm}
\begin{corollary}\label{cor:arbhard}
Given any monotonic operator $\op$, input $I$ of size $N$, output $R$ of size $M$, and an output $r\in R$, it is \#P-complete in $N$ and $M$ to compute $\Pimp$.
\end{corollary}
\vspace{-2mm}
}
\subsubsection{Bounded-size \miset s}
\label{subsubsec:bdd}
Given the intractability of the general problem for $\Pall$ and $\Pimp$ for \manyone\ and \arbitrary\ monotonic operators, we explore an intuitive tractable subclass. We consider a practical special case of all \miset s being of small size (i.e., bounded by a constant). The following theorem shows that we can now infer all types of $\P$ in polynomial time, using an explicit search.
\begin{theorem}\label{thm:bounded}
Given any monotonic operator $\op$, input $I$, and output $r\in R$ we can find each of $\Pall$, $\Pany$, $\Punion$, $\Pint$, and $\Pimp$ for $r$ in $\sim N^B$, when for every $S\in \Mall(r)$, we have $|S|\leq B$.
\end{theorem}

%% file: results-algos-complexity.tex
\begin{table}[t!]
\begin{center}
\scriptsize
\begin{tabular}{|c| c c c c|} 
\toprule
{\bf Properties} & {\bf $\Pall$,$\Pimp$} & {\bf $\Pany$}& {\bf $\Punion$}& {\bf $\Pint$} \\ 
\midrule
{\bf \arbitrary} & \multirow{2}{*}{\#P-complete$^\ddagger$} & \multirow{2}{*}{${\cal O}(MN\alpha^{k})^\dagger$} & \multirow{2}{*}{${\cal O}(2^N)^\ddagger$} & \multirow{2}{*}{${\cal O}(MN)$} \\
{\bf \manyone} & & & & \\ \midrule
{\bf \onemany} & ${\cal O}(M$+$N)$ & ${\cal O}(M$+$N)$ & ${\cal O}(M$+$N)$ & ${\cal O}(M$+$N)$ \\ \midrule
{\bf \oneone} & \multirow{2}{*}{${\cal O}(N)$} & \multirow{2}{*}{${\cal O}(N)$} & \multirow{2}{*}{${\cal O}(N)$} & \multirow{2}{*}{${\cal O}(N)$}  \\
{\bf {\tt unique \miset}} & & & & \\
\bottomrule
\end{tabular}
\vspace{-4mm}
\caption{\label{tab:summary}\scriptsize Complexity of our algorithms for obtaining $\Pall$, $\Pany$, $\Punion$, $\Pint$, and $\Pimp$ for various properties of an operator $\op$ with input $I$ and output $O$ ($N=|I|$ and $M=|O|$). $^\dagger$ denotes the number of \miset s that need to be found, and $\alpha$ is the size of the largest \miset\ ($\alpha\leq N$). $^\ddagger$ denotes cases where complexity becomes PTIME when restricted to bounded-size \miset s.}
\vspace{-6mm}
\end{center}
\end{table}

%% file: consolidation.tex
\section{Putting it all together} 
\label{sec:alltogether}

\subsection{Composing Operator Provenance}
\label{subsec:composition}
\label{sec:composition}

So far, we focused only on computing provenance for a single operator. We now consider composition of provenance from single operators into provenance for a chain of operators. Our goal is to understand to what extent (if at all) we can use each individual operators' provenance to determine the provenance of a pipeline. Formally, we would like to solve the following problem.
\begin{problem}\label{prob:composition}
Given monotonic operators $O_1,O_2$, input $I_1$, outputs $R_1=O_1(I_1)$ and $R_2 = O_2(R_1)$, and $r_2\in R_2$, can we compute $\P_*(O_2\circ O_1,I_1,r_2\in R_2)$ from $\P_*(O_2,R_1,r_2\in R_2)$ and $\P_*(O_1,I_1,r_1\in R_1)$.
\end{problem}
\noindent Intuitively, we are interested in generating all provenance $\P_*$ for the composition operator $O_{12} = (O_2\circ O_1)$ from all the provenance $\P_*$ of each of $O_1$ and $O_2$. Before proceeding to solve the above problem, we make two observations. First, note that Problem~\ref{prob:composition} could have been equivalently defined if $O_1$ and $O_2$ were themselves chains of operations with $\P_*$ being the provenance of these chains. Our algorithms in Section~\ref{subsec:provconstruction}, and hence results in this section, make no assumption on $O_1$ and $O_2$ being single operators, so all our results carry over when they are chains of operators. Second, our goal is to explore to what extent the provenance of $O_1$ and $O_2$ can be reused to generate the provenance of $O_{12}$, without any additional execution of $O_1$ or $O_2$. 

Our first main result shows that the execution of $O_2\circ O_1$ can be completely simulated using $\Pall$ for $O_1$ and $O_2$, and hence all provenance of $O_2\circ O_2$ can be computed as in Section~\ref{subsec:provconstruction}. The theorem gives a constructive algorithm, and hinges on the core idea that $\Pall$ for any monotonic operator captures enough information to execute the operator on any subset of its input.
\begin{theorem}\label{thm:pall-cons}
Given monotonic operator $O_1,O_2$, input $I_1$, outputs $R_1=O_1(I_1)$ and $R_2 = O_2(R_1)$, and $r_2\in R_2$, for any $I_s\subseteq I_1$, we have $r_2\in (O_2\circ O_1)(I_s)$ if and only if $\exists M_2\in \Pall(O_2,R_1,r_2)$ such that $M_2\subseteq O_s$, where $O_s = \{r_1 | \exists M\in \Pall(O_1,I_1,r_1\in R_1)\ s.t.\ M\subseteq I_s\}$.
\end{theorem}
\noindent Given the above result, we know that $\Pall$ for $O_1$ and $O_2$ contain enough information to compute all $\P^*$ for $O_2\circ O_1$. However, using $\Pall$ can be expensive because of the number of possible \miset s. Hence, our next goal is to attempt to use other forms of provenance of $O_1$ and $O_2$, i.e., without enumerating $\Pall$. If some provenance of $O_{12}$ cannot be computed directly, we can fall back on the techniques from Section~\ref{subsec:provconstruction} to generate the provenance or use $\Pall$ using Theorem~\ref{thm:pall-cons}. 

Next we look at special cases of operators and determine when $P_*$ for $O_2\circ O_1$ can be computed efficiently. Table~\ref{tab:cons} summarizes our results. The table is {\em complete} in the sense that for any $P_*$ not present in the table, we must use $\Pall$ for $O_1$ and $O_2$ (using Theorem~\ref{thm:pall-cons}) to compute the entry, or resort to techniques in Section~\ref{subsec:provconstruction}. Given all possible combinations of operator properties is too many to list (16 combinations of arbitrary, \manyone, \oneone, and \onemany), Table~\ref{tab:cons} presents a delineating subset of results. All other combinations of results can be derived from the entries in the table. For instance, when $O_2$ is \oneone\ $\P_*^{12}$  can be computed for arbitrary $O_1$, hence other combinations involving $O_1$ aren't present in the table. Similarly, we only consider arbitrary $O_2$ when $O_1$ is arbitrary. Further, we do not consider \manyone\ separately as results for \manyone\ and arbitrary are similar. Also, results for \oneone, \onemany\ are similar as they both ensure a unique singleton \miset\ for each output record. We don't separately consider unique \miset s, as all of $\Pall$, $\Pany$, $\Punion$, $\Pint$ are the same, and computed easily. Finally, the table omits $\Pimp$ as our solution for $\Pimp$ is equivalent to that of using $\Pall$.

\begin{table}[t!]
\scriptsize
\begin{center}
\begin{tabular}{|c|c|} \hline

{\bf Properties} & {\bf $\P_*(O_2\circ O_1,)$} \\ \hline

{\bf $O_1$: \arbitrary} &  ${\Punion}^{12}(r_2) \subseteq \bigcup_{r_1\in {\Punion}^2(r_2)}  ({\Punion}^1(r_1))$   \\ 

{\bf $O_2$: \arbitrary} & ${\Pint}^{12}(r_2) \supseteq \bigcup_{r_1\in {\Pint}^2(r_2)}  ({\Pint}^1(r_1))$  \\ \hline

{\bf $O_1$: \oneone} &   ${\Pall}^{12}(r_2) = \{ \bigcup_{s1\in s2} {\Pany}^1(s_1) | s2\in {\Pall}^{2}(r_2)\}$ \\

{\bf $O_2$: \arbitrary} & $\P_*^{12}(r_2) = \bigcup_{s1\in \P_*^{2}(r_2)} {\Pany}^1(s_1)$$^\dagger$\\ \hline

{\bf $O_1$: \arbitrary}  & \multirow{2}{*}{$\P_*^{12}(r_2) = \P_*^{1}(\P_*^2(r_2))$$^\ddagger$} \\ 

 {\bf $O_2$: \oneone}  & \\ \hline

\end{tabular}
\vspace{-3mm}
\caption{\label{tab:cons}\scriptsize Problem~\ref{prob:composition} for combinations of properties for $O_1$ and $O_2$ operating on inputs $I_1$ and $R_1=O_1(I_1)$ respectively, and $R_2$ is the result of $O_2$. The table uses $\P_*^{12}(r_2)$, $\P_*^{2}(r_2)$, and $\P_*^{1}(r_1)$ as shorthands for $\P_*(O_2\circ O_1,I_1,r_2\in R_2)$, $\P_*(O_2,R_1,r_2\in R_2)$ and $\P_*(O_1,I_1,r_1\in R_1)$ respectively. $^\dagger$ * stands for one of {\em uni}, {\em int}, and {\em any}. $^\ddagger$ We have slightly abused notation to apply $\P_*^1$ to a singleton set, instead of the record in the set itself.}
\vspace{-4mm}
\end{center}
\end{table}

\subsection{Properties and Provenance Selection}

To summarize, given an IE execution our approach, \prober, allows users to specify the output records that they are interested in debugging. Either using information such as IO specifications or integrity constraints or using sampling, \prober\ attempts to identify the type of the operators (e.g., \oneone, or arbitrary). In the absence of any conclusive information, \prober\ assumes an arbitrary operator. For each operator or pipeline, users may choose the type of provenance they want based on editorial resources available. Note that users may always start with the conservative $\Pint$  or $\Pany$, and explore more complex provenance, such as $\Pall$ as needed, or ask for input to be ranked, such as $\Pimp$. We make two important observations regarding this user exploration: (1) Whenever operators satisfy restricted properties (such as \oneone), \prober\ readily computes all forms of provenance very efficiently. (2) For arbitrary monotonic operators, all our algorithms proceed in a ``pay-as-you-go fashion''; for instance, even if a user would like to perform an in-depth analysis of $\Pall$ leading to a potentially expensive computation, \prober\ starts returning \miset s immediately and progressively provides more information as available. Specifically, our $\Pany$ algorithm keeps iteratively finding new \miset s, which are returned to users as found.

%% file: 6-experiments.tex
\section{Experimental Evaluation}
\label{sec:experiments}
We now present results from our experimental evaluation. After describing our data sets (Section~\ref{sec:settings}), we present a qualitative study of \prober\ (Section~\ref{sec:quali}). Next, we evaluate the effectiveness of our basic unit for provenance, namely, \miset s (Section~\ref{sec:effect}). Then, we perform a detailed study of various provenance formalisms (e.g., $\Punion, \Pint, \Pany,$ etc.) by discussing basic statistics (Section~\ref{sec:stats}), and then, compare their coverage (Section~\ref{sec:coverage}) and execution times (Section~\ref{sec:time}).

\subsection{Experimental Settings}
\label{sec:settings}

\paragraph{Data sources:} We used a collection of 500 million web pages crawled by the Yahoo! search engine.

\paragraph{Extraction pipelines:} For our IE tasks, we implemented two pipelines. Our first pipeline, denoted, \grush, is as described in Section~\ref{sec:motivation}. For our second pipeline, denoted \wpr, we reimplemented a state-of-the-art bootstrapping exraction technique described by Pasca et al.~\cite{aaai06/pasca} for large-scale datasets such as Web corpora which is similar in spirit to other IE pipelines such as Snowball~\cite{dl00/agichtein} and Espresso~\cite{acl06/pantel}.

\paragraph{Extracted relations:} As extraction tasks, we focus on six relations (the last column shows the number of extracted tuples):

\vspace{-2mm}
\begin{table}[h]
\centering
\scriptsize
\begin{tabular}{ll ll}
{1} & {\bf \bussi}: & $\langle$name, address, phone$\rangle$ & 340,131\\
{2} &{\bf \actors}: &  $\langle$movie, actor$\rangle$ & 14,414\\
{3} &{\bf \books}: &   $\langle$book, author$\rangle$ & 142,337\\
{4} &{\bf \mayor}: &  $\langle$U.S. city, mayor$\rangle$ & 28514\\
{5} &{\bf \senparty}: &  $\langle$senator, affiliated party$\rangle$ & 2,119\\
{6} &{\bf \senstate}: &  $\langle$senator, state$\rangle$ & 14,582
\end{tabular}
\vspace{-2mm}
\end{table}

We built \bussi\ using \grush\ using a corpus of 5,443,183 web pages from 147 sites (see Section~\ref{sec:motivation}); all the other relations were built using \wpr. Our qualitative analysis presented next is using the \bussi\ dataset. The empirical analysis that follows was performed on each of \actors, \books, \mayor, \senparty, and \senstate. For most of our experiments we show results for the high-confidence tuples from our datasets, as these results are the most interesting: High-confidence tuples have most number of contributing input records and are therefore the hardest for provenance and debugging. 

\subsection{Qualitative analysis of \prober's utility}
\label{sec:quali}
To gain insights into the utility of \prober, we performed a qualitative analysis of the records generated for \bussi. Among the final set of output records, 38\% were missing business names, 40\% were missing phone numbers, 37\% were missing addresses. To give a flavour of user interaction with \prober, we qualitatively depict a debugging analysis for a specific erroneous record. In particular, we explore a record, $r$, $\langle$`AUSTIN, TX', `Burnett St, Austin, Texas 78703', null$\rangle$ which has incorrect values for business name and a missing value for phone number. Through the source web page associated with this record, we found that our first operator, namely \seg, had incorrectly segmented the page. As shown in Figure~\ref{fig:ex2}, \seg\ generated an incorrect segmentation for the second and third business contacts listed on the page. By fixing this segmentation, we debug and correct record $r$ as well as other records extracted from this page.

\begin{figure}[t]
\centering
\includegraphics[height=2in]{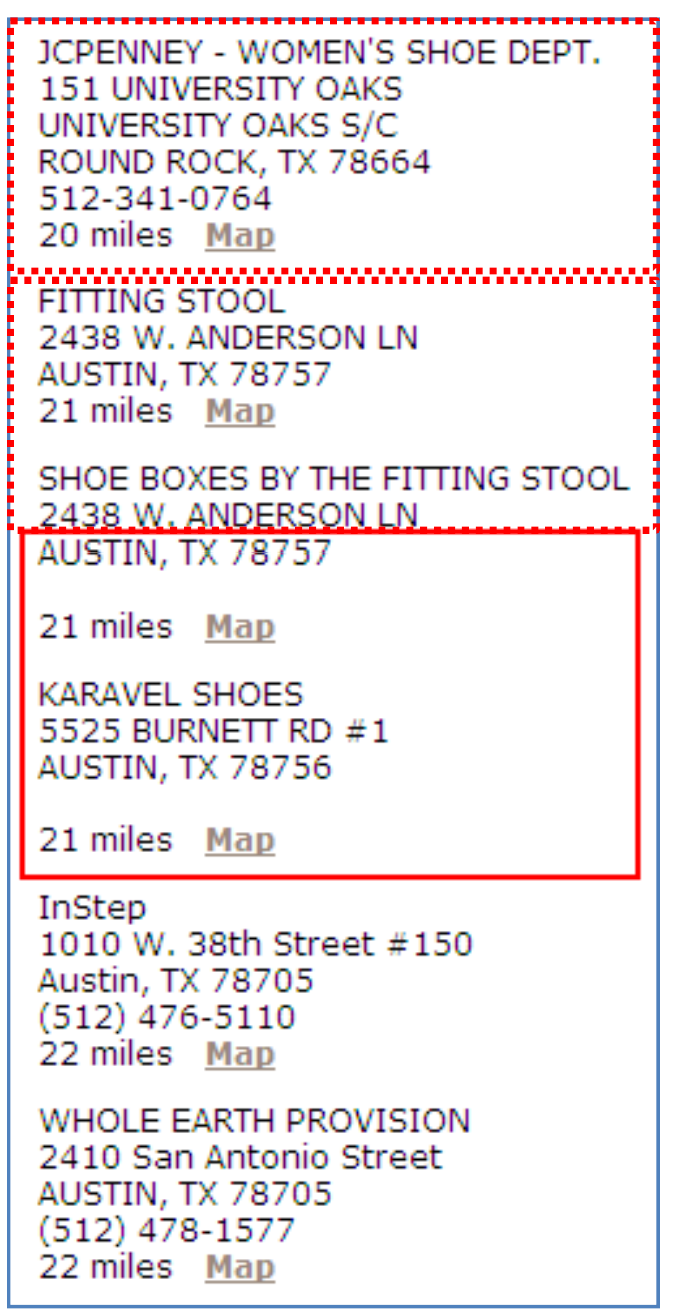}
\vspace{-5mm}
\caption{Incorrect segmentation causing incorrect output.}
\vspace{-3mm}
\label{fig:ex2}
\end{figure}

\subsection{Is \miset\ an effective representation?}
\label{sec:effect}

Earlier in Section~\ref{sec:prov}, we proposed \miset s as the primary representation to collect information related to an output record for debugging purposes, and provided concrete theoretical justification for our choice. Of course, other representations are also possible in practice. Next, we present an experimental comparison of \miset s against three strong baselines that could be used to collect tracing information for an output record. Specifically, we compare the following methods for generating tracing information.
\vspace{-1mm}
\begin{itemize}\itemsep-0.03in
\item \alldocs: The naive baseline of repeatedly exploring all input records for every output record. 
\item \wordor: Using the bag of words in an output record, we build a keyword query to fetch all input documents containing at least one term using a standard IR-like search interface.
\item \wordand: Similar to \wordor, except we only fetch documents containing {\em all} terms in the output record.
\end{itemize}
\vspace{-1mm}

\noindent An important note about the \wordor\ and \wordand\  baselines is that they exploit specific information about the extraction operators, namely, that input records aren't ``mangled'', i.e., terms are preserved by the extraction. \miset s, on the other hand, use no such information. Since in our extraction scenario, we chose operators that do indeed preserve terms in records, our comparison is unfair in that it favors \wordor\ and \wordand. Our goal was to compare \miset s with the {\em best} possible scenario for our baselines. (Clearly, in a fair comparison including operators that generate new terms or alter terms in input records, \wordor\ and \wordand\ won't even be applicable, and \alldocs\ would be the only feasible baseline.)

Figure~\ref{fig:nr-docs} presents our results comparing each method (\miset s, \wordand, \wordor, \alldocs) by examining the total number of input records that need to be fetched in order to generate the tracing information, varying the number of output records. By design, \wordand\ retrieves the fewest possible documents and \miset s completely coincides with \wordand. Indeed, this ``experimentally proves'' our claim from Section~\ref{sec:prov} that \miset s retrieve minimal sets of records from the input. Note that even in our favorable setting for keyword-based retrieval, \wordor\ retrieves many more input records\footnote{Note the log-scale on the y-axis}, and \alldocs\ is even more prohibitively expensive.

\subsection{Size of provenance formalisms}
\label{sec:stats}

\begin{figure*}[t]
\centering
\vspace{-5mm}
\begin{tabular}{c c c} 
\subfigure{\includegraphics[width=0.6\columnwidth]{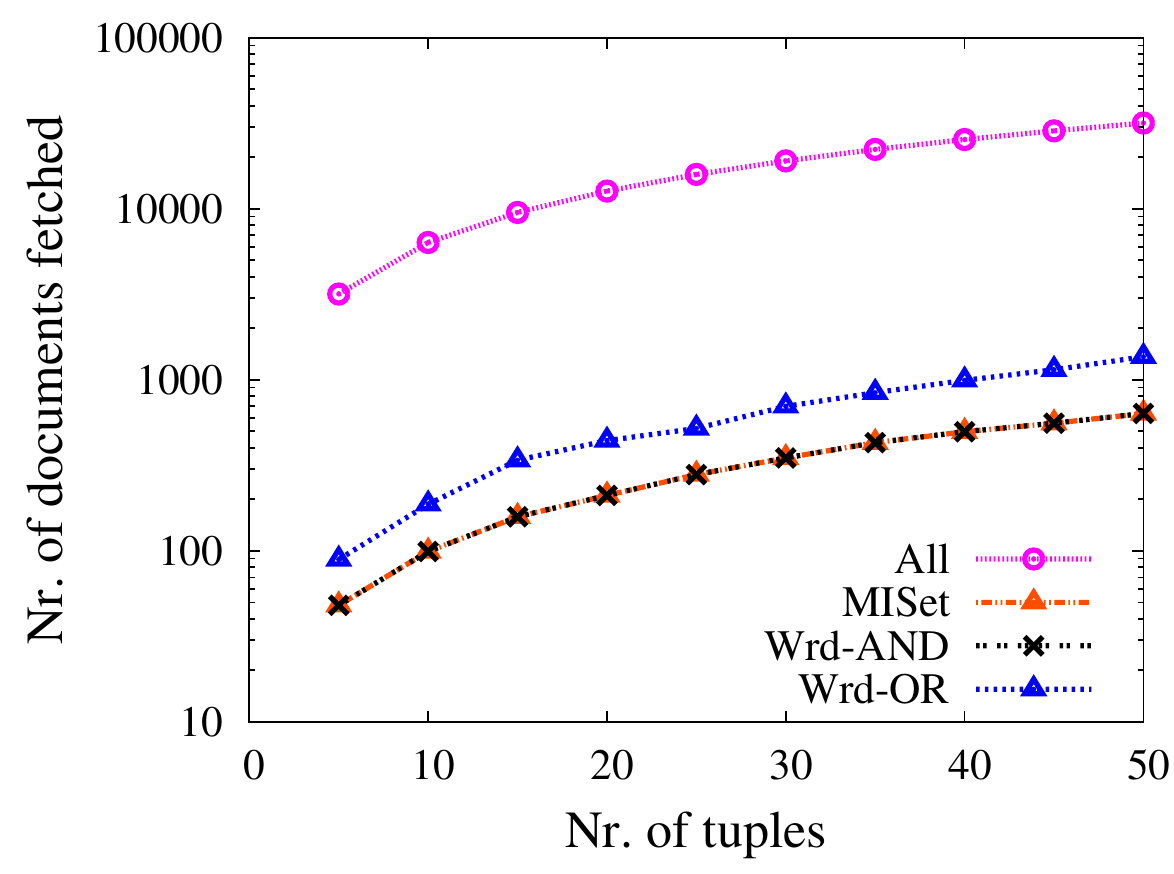}\label{fig:nr-docs}} &
\subfigure{\includegraphics[width=0.6\columnwidth]{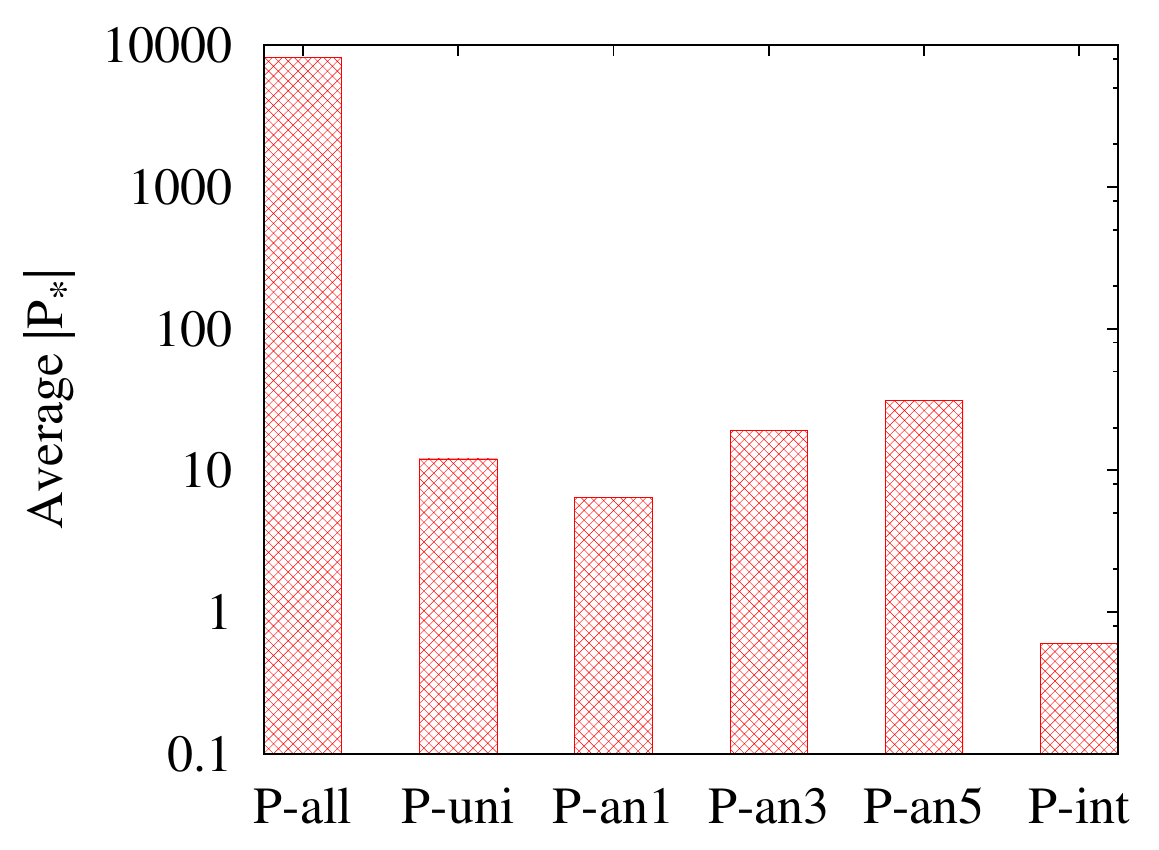}\label{fig:avg-sizes}} & 
\subfigure{\includegraphics[width=0.6\columnwidth]{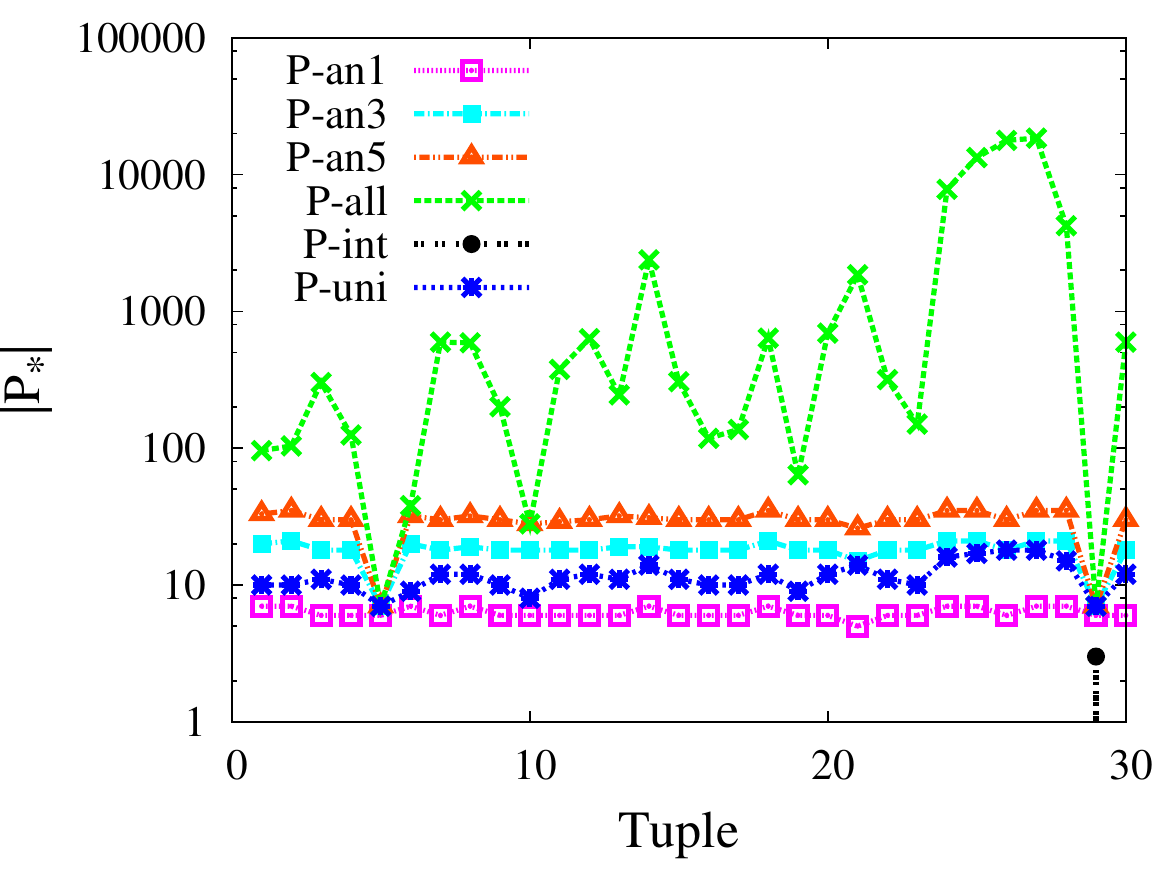}\label{fig:sizes}} \\[-2mm]
{\bf (a)} & {\bf (b) } & {\bf (c)} \\[-2mm]
\subfigure{\includegraphics[width=0.6\columnwidth]{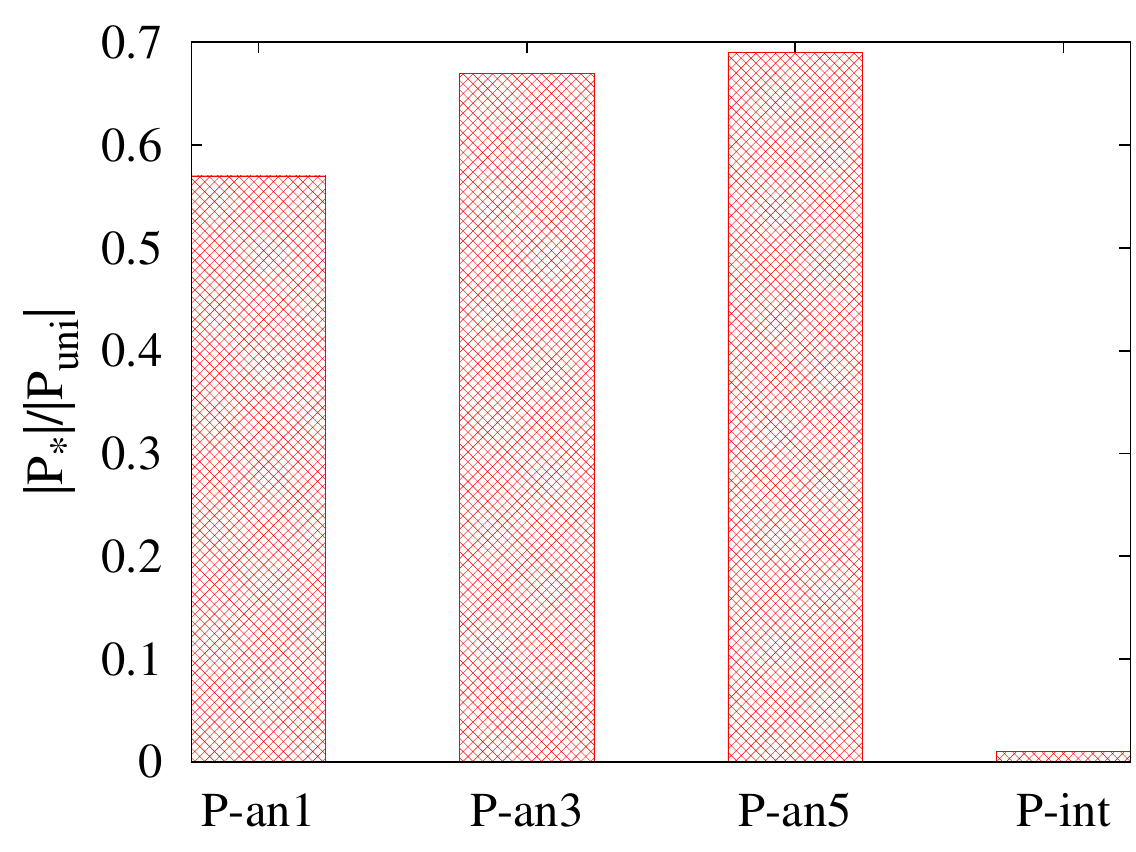}\label{fig:coverage}} & 
\subfigure{\includegraphics[width=0.6\columnwidth]{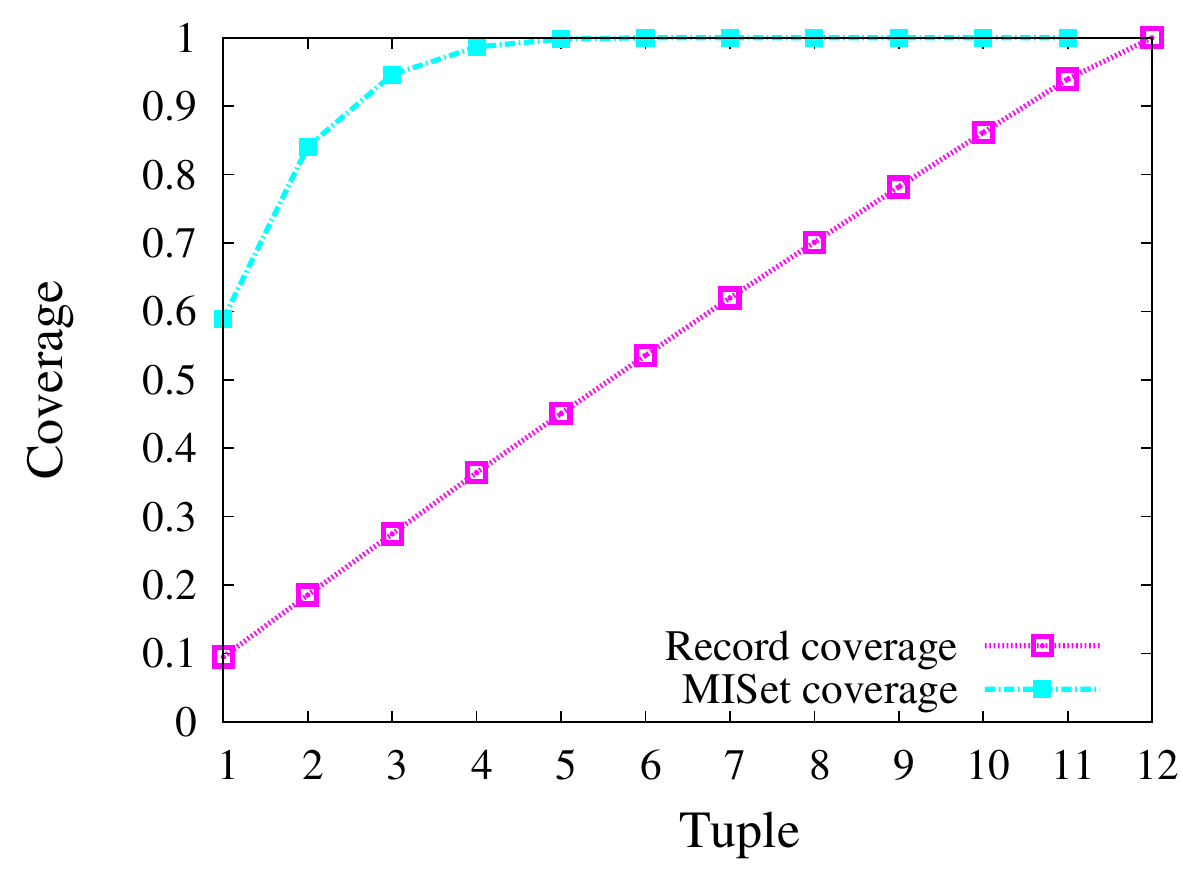}\label{fig:pimp-coverage}} & 
\subfigure{\includegraphics[width=0.6\columnwidth]{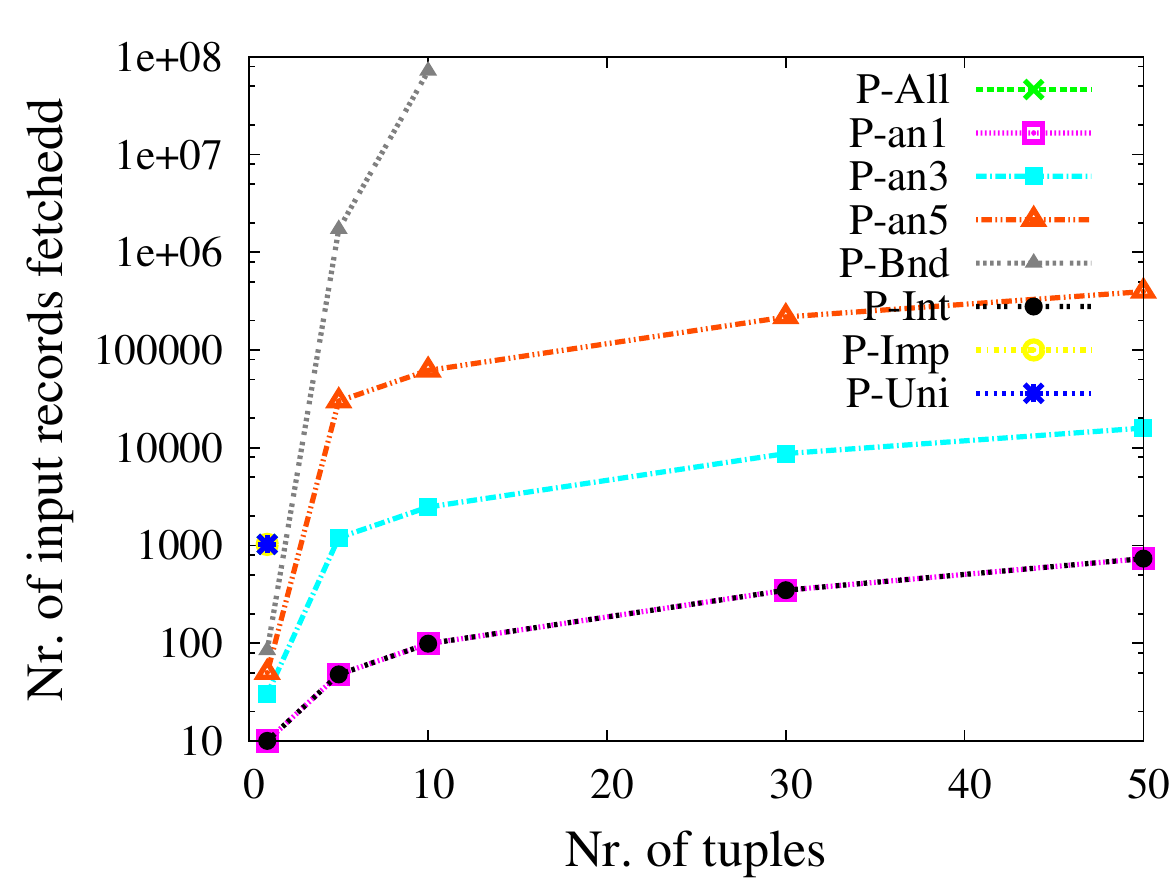}\label{fig:nr-execs}} \\[-2mm]
{\bf (d)} & {\bf (e) } & {\bf (f)} \\
\end{tabular}
\vspace{-2mm}
\caption{\scriptsize \bf (a) Number of documents fetched representing the amount of work necessary when using different debugging paradigms. (b) Average size of various provenances over 50 tuples. (c) Size of provenance generated for top-30 tuples ranked by confidence scores. (d) Coverage of different provenances with respect to $\Punion$. (e) Coverage of $\Pimp$ with respect to total number of \miset s, and total contribution of all input records. (f) Data fetched to derive various provenance $P_*$ definitions.}
\label{fig:massive}
\vspace{-3mm}
\end{figure*}

Next we explore the size of provenance generated using each of our formalisms: $\Pall, \Punion, \Pint$, and $\Pany$ with $k=1,3,5$. ($\Pimp$ isn't shown as the size of $\Pimp$ is naturally equal to the number of input records requested.) Figure~\ref{fig:avg-sizes} shows the average size of the provenance generated for each of the provenance formalisms over a set of 50 tuples ranked by their confidence scores. It is noteworthy that the sizes of the provenances, and in turn, the manual effort necessary can substantially vary across tuples. Since $\Pall$ maintains all possible \miset s, it is the largest. From the figure, we learn that a practical choice for users would be to start exploring $\Pint$ or $\Pany$-1, then request $\Pany$-$k$ for $k>1$ and $\Punion$ if necessary. 

To gain more insight into the distribution of sizes for individual tuples, Figure~\ref{fig:sizes} plots the size of each provenance type for the top-30 tuples. The size of $\Pall$ varies significantly but is almost always significantly more than all other provenance types. The two cases where $\Pall$ coincides with other provenance types  are examples of output records with unique \miset s. The minor variations in the sizes of all other forms of provenance are obscured by the log-scale for the y-axis.

\subsection{Coverage}
\label{sec:coverage}

Next we explore the  {\em coverage} of each provenance model measured as $\frac{|P_*|}{|\Punion|}$. Our goal is to determine what fraction of all potentially contributing input records is retrieved by any single \miset\ or any arbitrary 3 or 5 \miset s, as well as by $\Pint$. Figure~\ref{fig:coverage} shows the coverage averaged over the set of output records. $\Pint$ has very low coverage indicating that very few input records are {\em essential} in producing any output record; in other words, in most cases there are many different explanations for the same output record. $\Pany$ (for $k=1,2,3$), on the other hand, retrieves a sizable fraction of all contributing input records. This indicates that using $\Pany$ is a practical solution to start debugging, by retrieving the initial set of input records, and if necessary request more \miset s.

We treat the coverage study for $\Pimp$ as a special case. Since the coverage of $\Pimp$ depends on the number of ranked tuples retrieved, we measure the coverage of top-$k$ $\Pimp$ records $\{r_1, \ldots, r_k\}$ using two measures: (1) Record-coverage measured as the fraction of the total number of record appearances of these $k$ records in $\Pall$. That is $\frac{\sum_{i=1}^{k} c_i}{\sum_{i=1}^{l} c_i}$, where $c_i$ denotes the number of \miset s containing $r_i$, and $\Pall$ contains records $\{r_1, \ldots, r_l\}$. (2) \miset-coverage measuring the fraction of the total number of \miset s that contain some tuple in $\{r_1, \ldots, r_k\}$. Figure~\ref{fig:pimp-coverage} shows these coverages for $\Pimp$; we observe that $\Pimp$ is very effective in giving very high \miset-coverage with very few retrieved records, justifying that retrieving few tuples from $\Pimp$ can be very useful in debugging with a high representation of almost all \miset s. We get high incremental value for initial records, with diminishing returns as we retrieve more tuples. For record-coverage the trend is closer to a linear increase in coverage. Overall, we observe that $\Pimp$ (along with with $\Pany$ and $\Pint$) can be effective tools for debugging, with the caveat that $\Pimp$ is computationally more expensive (see Section~\ref{sec:time}). An interesting open question arising is that of efficiently (to the extent possible, given our \#P-complete result from Section~\ref{sec:algos}) retrieving just sufficient number of records to meet a coverage demand.

\subsection{Build time}
\label{sec:time}

Finally, we study the time required to build provenance in \prober, which directly depends on the amount of data fetched. Figure~\ref{fig:nr-execs} plots the number of input records fetched for each type of provenance, varying the number of high-confidence records. 
We note that $\Pall$, $\Punion$, and $\Pimp$ are the most expensive computationally, while the amount of data fetched for $\Pint$ and $\Pany$ for $k=1,3,5$ is significantly less. Since $\Pall$ requires a large amount of data to be fetched, we studied the behavior of our algorithm for finding all \miset s when the size of each \miset\ is bounded below 5 (Section~\ref{subsubsec:bdd}). We notice that this is more expensive than $\Pany$ and $\Pint$ but significantly faster than $\Pall$, and hence information on the size of each \miset\ can potentially be useful.

\subsection{Evaluation summary} In conclusion, we established the utility of \prober\ over a variety of relations. \miset s pick out minimal sets of input records in comparison to other baseline methods thus enabling rapid resolution of output records. Our provenance formalisms may substantially vary in their sizes and we discussed how users may gradually move from exploratory provenances to more complete ones. Finally, we studied the tradeoff between coverage and execution time for various provenance formalisms.

%% file: 7-conclusion.tex
\section{Related Work and Conclusions}
\label{sec:relatedwork}

Here we present a very brief discussion of related work, with a more comprehensive description appearing in Appendix~\ref{app:relatedwork}. Some recent work~\cite{icde08/jain,nyu07/jain,joins/jain,goodk/jain,sigmod08/shen} has broadly looked at providing {\em exploration} phases that enable users to determine if a text database is appropriate for an IE task. However, users are provided with little or no insight into why unexpected results are produced, and how to debug them. Another interesting piece of work~\cite{vldb07/shen} presented techniques to build IE programs using Datalog for greater readability and easier debugging. Our recent work~\cite{sigmod10/sarma} considered debugging for iterative IE, and~\cite{nonanswer} looked at provenance for {\em non-answers} in results of extracted data. However, these papers assume complete knowledge of each operator in some form, such as access to code for each operator, or SQL queries applied to input data. Finally, there is a large body of work on provenance for relational data (refer~\cite{ikeda-survey,provenance}), and more recently~\cite{chap-jag10} on understanding provenance information. This work does not address the problem of building provenance for black-box operators to facilitate IE debugging with minimal editorial effort, the primary goal of our work.

In conclusion, we presented \prober, the first system for ad-hoc debugging of IE pipelines. At the core of \prober, is a suitable \miset-based provenance-model to link each output record with a {\em minimal} set of contributing input records. We provided efficient algorithms and complexity results for provenance inference, and an extensive experimental evaluation on several real-world data sets demonstrating the effectiveness of \prober. A few specific directions for future work arise, such as tighter bounds for $\Punion$ inference, and extending to non-monotonic operators. A more general direction of future work we are currently pursuing is to develop an interactive GUI for \prober\ and perform a user study by deploying it for multiple extraction frameworks at Yahoo!.

%% file: appendix.tex
\pagebreak

\appendix

\begin{figure*}[t]
\centering
\includegraphics[width=2\columnwidth]{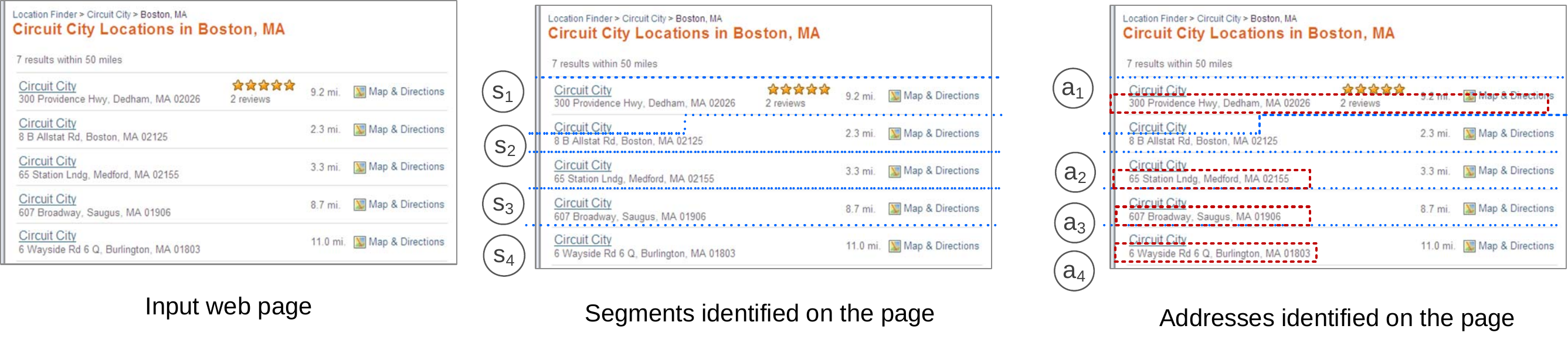}
\vspace{-5mm}
\caption{Sample input output for three steps, namely, \seg, \add, and \phn, in our extraction pipeline.}
\vspace{-3mm}
\label{fig:sample-data}

\end{figure*}

\input{proofs.tex}

\section{Expanded Related Work}
\label{app:relatedwork}

Recognizing the need for a principled approach to assisting IE developers and users, several methods have been proposed to enable {\em exploration} phases. Shen et al.~\cite{sigmod08/shen} presented an iterative approach to developing IE systems, where users begin with an ``approximate'' extraction query. Based on the results of this query, users may refine the follow-up query. Jain et al.~\cite{goodk/jain} presented a query model for IE tasks for the purpose of exploring whether a database is useful for the IE task or not. Following a similar spirit, optimization strategies that enable users to efficiently fetch IE results with pre-specified output quality (e.g., minimum number of good tuples and maximum tolerable bad tuples) have been proposed for single IE systems~\cite{icde08/jain,nyu07/jain} as well as multiple IE systems~\cite{joins/jain}. While such exploration phases enable IE developers to assess the quality of an IE system, they mostly focus on answering the question, ``{\em Is a text database $D$ a good choice for the IE system at hand?}'' Furthermore, users are provided with very little insights into why an IE system does not perform as expected. 

Assuming full access and control to the code for each operator in an IE pipeline, prior work~\cite{vldb07/shen} presented methods to build IE programs involving multiple operators using Datalog to generate programs that are easy to read and thus easier to debug. Our approach considers a generic IE pipeline that may involve any arbitrary operators for which we may not have exact specifications or access to the code. Recently,~\cite{chap-jag10} addressed the problem of understanding ``provenance black boxes''; the goal of their work is to provide users with an easier way to understand provenance information, allowing them to aggregate or drill down on provenance. In contrast, our goal is to build a provenance model that is suitable for black-box operators in an extraction pipeline. Note that we make no direct contribution on user understanding of provenance; rather, we produce {\em minimal sets} of provenance and their compositions in order to quickly understand errors in the data produced by the pipeline.

Our prior work~\cite{sigmod10/sarma} presented debugging algorithms for IE tasks; our current work substantially differs from and extends this work. The techniques in~\cite{sigmod10/sarma} focused on a simple form of IE system, namely, iterative IE methods~\cite{acl06/pasca}. Furthermore, we assumed we had complete knowledge on how each operator was designed and exact operator and input-output specifications . Moreover, the focus of~\cite{sigmod10/sarma} was to utilize the relatively simple provenance model to enable efficient algorithms for explanation, diagnosis, and repair. This paper developed a new provenance model for arbitrary extraction operators, and presented algorithms for building this provenance. Also relevant previous work on provenance is that of~\cite{nonanswer}, which addresses the problem of deriving the provenance (explanations) for {\em non-answers} in extracted data. The paper considers conjunctive queries, and for every potential tuple $t$ in an answer to a conjunctive query, the authors provide techniques for determining updates to base data that would produce $t$ in the output. Once again, our work relaxes these assumptions and enables debugging over complex IE pipelines consisting of arbitrary black-box operators. Finally, there is a large body of previous work on provenance for relational databases (refer to~\cite{provenance,ikeda-survey} for surveys); this work does not meet our two-fold requirements of provenance for black-box operators, and designing provenance to minimize editorial work during debugging.


\input{4-1-discussion.tex}


%% file: proofs.tex
\section{Proofs}
\label{sec:proofs}

\paragraph{Proofs of Theorem~\ref{thm:umset}, Lemma~\ref{lem:utest}, and Lemma~\ref{lem:anyset}:}
To prove Lemma~\ref{lem:utest}, consider Algorithm~\ref{algo:unique}, which attempts to find non-uniqeueness of \miset s, starting from a given \miset\ $M$ obtained from Algorithm~\ref{algo:anyset}. Any other \miset\ $M'$ cannot be a superset of $M$ (else it wouldn't be minimal). Therefore, there must exists some $m\in M$ and $m\not \in M'$, which implies that $I-\{m\}\supseteq M'$. Therefore $O(I-\{m\})$ must contain $r$ for some $m\in M$ if there exists a \miset\ other than $M$.

The basic algorithm for Lemma~\ref{lem:anyset} is Algorithm~\ref{algo:anyset} which finds any \miset\ for a given input $I$. Finding $k$ \miset s is simply obtained by modifying input $I$ and calling Algorithm~\ref{algo:anyset} recursively: To find an \miset\ different from $M$, for each element $m\in M$, Algorithm~\ref{algo:anyset} is called with $I-\{m\}$. Similarly, given $p$ \miset s $M_1, \ldots, M_p$, to find a $(p+1)$th \miset\ $M'$, $M'$ must differ from each of $M_1, \ldots, M_p$. Hence, there must exist a $p$-tuple $(m_1,\ldots, m_p)$, $m_i\in M_i$, such that $M' \subseteq I-\{m_1, \ldots,m_p\}$. Our algorithm attempts to find an \miset\ for every such $p$-tuple. Finding the $(p+1)$th \miset\ needs to iterate over $|M_1|\cdot \ldots \cdot |M_p|$ $p$-tuples in the worst case, giving us the required complexity.

Theorem~\ref{thm:umset} follows based on checking whether $O$ gives a unique \miset\ (Lemma~\ref{lem:utest}), then using Lemma~\ref{lem:anyset} with $k=1$.\rbox

\paragraph{Proof of Theorem~\ref{thm:111nall}:}
We scan the input records $i\in I$, one at a time, and apply $\op$ to $\{i\}$ individually. Whenever we have $\op(\{i\})$, we return $\Pany(o)=\{i\}$, and add $\{i\}$ to $\Pall(o)$, and add the element $i$ to $\Punion(o)$, initialized to $\emptyset$. If $|\Punion(o)|\geq 2$, we set $\Pimp(o)=\emptyset$ , else set $\Pimp(o)=\Punion(o)$. Finally, $\Pint$ can be computed from $\Punion$.

It can be seen easily that provenance for \onemany\ operators can be computed in a similar fashion. The only difference is that the number of output records can now be larger than the number of input records, i.e., $M$ may be larger than $N$. Hence the complexity increases to ${\cal O}(M+N)$.\rbox

\paragraph{Proof of Theorem~\ref{thm:pint}:}
We perform ${\cal O}(N)$ calls to the operator, and for each call, we may have to look at an output of size ${\cal O}(M)$. Correctness of the algorithm follows easily: For any record $i$ to be in the intersection of all \miset s, removing $i$ from the input must remove the output record.\rbox

\paragraph{Proof of Theorem~\ref{thm:n1hard}:}
First we prove \#P-hardness for a \manyone\ operator (and hence for an \arbitrary\ operator), and then show that the problem is in $\#P$, which applies for \manyone\ and \arbitrary\ monotonic operators, completing our proof.

\begin{enumerate}
\item {\bf \#P-hardness} We give a reduction from the problem finding all minimal vertex covers. Given a graph $G(V,E)$, our goal is to compute all $V_{min}\subseteq V$ such that (1) $V_{min}$ is a {\em cover}: each edge $e\in E$ has an endpoint in $V_{min}$, (2) $V_{min}$ is {\em minimal}: No proper subset of $V_{min}$ is a cover. Given the input $G(V,E)$, we create an instance of finding $\Pall$ as follows: $I=V$, $O=\{1\}$, our goal is to find all \miset s of $1$. $\op$ takes as input any subset $I_s\subseteq I$, and returns $\{1\}$ if the corresponding set of vertices $V_s$ is a (not necessarily minimal) cover of $E$ in $G$, and returns $\{0\}$ otherwise. Note that each minimal vertex cover of $G$ corresponds to a \miset\ of $1$, and each \miset\ gives a minimal vertex cover. Finally, note that our operator generates a single output record, and is therefore \manyone.

\item {\bf \#P} Given any $I_s\subseteq I$, we can check in PTIME whether $I_s$ is an \miset: $I_s$ is an \miset\ if and only if no subset of it obtained by removing a single element returns $\{1\}$, and $1\in Op(I_s)$. Therefore, we can check for all sets in any $\Pall$, whether each of them is an \miset.\rbox
\end{enumerate}

\paragraph{Proof of Corollary~\ref{cor:n1hard}:}
Note that the hardness result of Theorem~\ref{thm:n1hard} holds even if our goal was to only count the number of minimal vertex covers, or equivalently, find the number of \miset s. We can translate the problem of counting the number of \miset s to computing $\Pimp$ for a special input tuple $i^*\in I$. Given an input $(Op,I,o\in O)$ to $\Pall$, we create $(Op',I',o\in O)$, where $I' = I\union \{i^*\}$ and for any $I_s\subseteq I'$ we have $Op'(I_s)=Op(I_s)$ if and only if $i^*\in I_s$ and $Op'(I_s)=\emptyset$ if $i^*\not\in I_s$. Counting the number of \miset s for $\op$ now reduces to the problem of determining $\Pimp$ for $i^*$.\rbox

\eat{
\paragraph{Proofs of Theorem~\ref{thm:arbhard} and Corollary~\ref{cor:arbhard}:}
\#P-hardness for Theorem~\ref{thm:arbhard} follows from Theorem~\ref{thm:n1hard}, and the proof of \#P-completeness is similar. Proof of Corollary~\ref{cor:arbhard} follows from Theorem~\ref{thm:arbhard}, as in the proof of Corollary~\ref{cor:n1hard}.\rbox
}

\paragraph{Proof of Theorem~\ref{thm:bounded}:}
The theorem follows directly based on an explicit search over all possible inputs of size of at most $B$ to find $\Pall$. All other $\P_*$ are subsequently computed using $\Pall$.\rbox

\paragraph{Proof of Theorem~\ref{thm:pall-cons}:}
The main idea used in the result is that the property of \miset s for monotonic operators ensures that $\forall r: \Pall(O,I,r\in R)$ for any operator is sufficient to reconstruct (and execute) $O$ for any subset $I_s\subseteq I$: Using monotonicity, we know that $O(I_s)\subseteq O(I)$, hence we only need to determine for every $r\in R$, whether $r\in O(I_s)$. Using the property of \miset s, we have $r\in O(I_s)$ if and only if there is a \miset\ of $r$, say $M_r\subseteq I_s$, allowing us to exactly construct $O(I_s)$. 

Given the above fact that $\Pall$ enables reconstructing any operator, the two expressions in the theorem merely simulate the execution of each operator: For a record $r_2$ to be in the output of $(O_2\circ O_1)$, some \miset\ $M_2$ of $r_2$ for $O_2$ must be contained in the output of $O_1$. Such an \miset\ $M_2$ is in the output of $O_1$, i.e., $M_2\subset O_s$.\rbox

%% file: 4-1-discussion.tex
\section{Extensions}
\label{app:extensions}

In this section, we very briefly discuss the extension of \prober s framework for absence of records from the output, which is particularly useful for non-monotonic operators. We emphasize that this section is primarily meant to indicate that \prober\ is amenable to these extensions. However, we are currently developing precise details, and our current system does not support these extensions. 

For debugging the absence of records from an output of any non-monotonic operator, we may analogously define a notion of {\em Maximal Superset} (\maset).

\begin{definition}[\maset]
Given an operator $O$, its input $I$ and output $R$, we say that  $I_s\subseteq I$ is a {\em Maximal Superset} (\maset) of $r\not \in R$ if and only if: (1)  $r\in Op(I_s)$; and (2) $\forall I': I'\supset I_s, I'\subseteq I \Rightarrow r\not\in Op(I')$.
\end{definition}
 
\noindent Just as in the case of \miset s, it's easy to see that \maset s are also not unique:

\begin{example}
In Example~\ref{ex:miset}, if the operator returned ``NO'' whenever there were fewer than 50 records in the input, then the \maset\ of ``NO'' is any set of 49 input records.
\end{example}

\eat{
\begin{theorem}
The union of all \miset s is equal to the intersection of all \maset s.
\end{theorem}
\begin{proof}
Is this true in general? Prove.
\end{proof}
}

\noindent 
\miset s are useful for debugging based on the {\em presence} records in the output of an operator, while \maset s are useful for debugging the {\em absence} of records from the output: For every record that is output by an operator, its \maset\ is the entire input, and hence its \maset\ doesn't help in fixing an erroneous output record. However, the \miset\ of an erroneous record points to potential incorrect input that caused the error. Conversely, for a record that is absent in the output, \maset s help identify what caused the record to get omitted from the output.

\eat{
\begin{example}
Consider an operator whose input is a set of relational tuples of schema $R(A,B)$. Suppose the operator performs a group-by $A$ and outputs the distinct sums of $B$ values. Given an input set of records $\{(1,1),(1,5),(2,4),(2,2)\}$, the output is a single record $(6)$, and it has two \miset s: $\{(1,1),(1,5)\}$ and $\{(2,4),(2,2)\}$.
\end{example}
}

\eat{
\subsection{To think}

Topics that we can add to extend this work
\begin{itemize}
\item \miset\ for a set of records, and not an individual record.
\end{itemize}
}

%% file: prober-arxiv.bbl
\begin{thebibliography}{10}

\bibitem{dl00/agichtein}
E.~Agichtein and L.~Gravano.
\newblock Snowball: Extracting relations from large plain-text collections.
\newblock In {\em DL}, 2000.

\bibitem{uldb06}
O.~Benjelloun, A.~{Das Sarma}, A.~Halevy, and J.~Widom.
\newblock {ULDBs}: Databases with uncertainty and lineage.
\newblock In {\em VLDB}, 2006.

\bibitem{chap-jag10}
A.~Chapman and H.~V. Jagadish.
\newblock Understanding provenance black boxes.
\newblock {\em Distributed and Parallel Databases}, 27(2), Apr. 2010.

\bibitem{nonanswer}
J.~Huang, T.~Chen, A.~Doan, and J.~F. Naughton.
\newblock On the provenance of non-answers to queries over extracted data.
\newblock {\em PVLDB}, 1(1), 2008.

\bibitem{ikeda-survey}
R.~Ikeda and J.~Widom.
\newblock Data lineage: A survey.
\newblock Technical report, Stanford University, 2009.

\bibitem{icde08/jain}
A.~Jain, A.~Doan, and L.~Gravano.
\newblock Optimizing {SQL} queries over text databases.
\newblock In {\em ICDE}, 2008.

\bibitem{nyu07/jain}
A.~Jain and P.~G. Ipeirotis.
\newblock A quality-aware optimizer for information extraction.
\newblock {\em ACM Transactions on Database Systems}, 2009.

\bibitem{joins/jain}
A.~Jain, P.~G. Ipeirotis, A.~Doan, and L.~Gravano.
\newblock Join optimization of information extraction output: Quality matters!
\newblock Technical Report CeDER-08-04, New York University, 2008.

\bibitem{goodk/jain}
A.~Jain and D.~Srivastava.
\newblock Exploring a few good tuples from text databases.
\newblock In {\em ICDE}, 2009.

\bibitem{cikm09/kasneci}
G.~Kasneci, S.~Elbassuoni, and G.~Weikum.
\newblock Ming: mining informative entity relationship subgraphs.
\newblock In {\em CIKM}, 2009.

\bibitem{acl06/pasca}
M.~Pa\c{s}ca, D.~Lin, J.~Bigham, A.~Lifchits, and A.~Jain.
\newblock Names and similarities on the web: Fact extraction in the fast lane.
\newblock In {\em Proceedings of ACL}, July 2006.

\bibitem{aaai06/pasca}
M.~Pa\c{s}ca, D.~Lin, J.~Bigham, A.~Lifchits, and A.~Jain.
\newblock Organizing and searching the world wide web of facts - step one: The
  one-million fact extraction challenge.
\newblock In {\em Proceedings of AAAI-06}, 2006.

\bibitem{acl06/pantel}
P.~Pantel and M.~Pennacchiotti.
\newblock Espresso: leveraging generic patterns for automatically harvesting
  semantic relations.
\newblock In {\em Proc. of ACL}, 2006.

\bibitem{sigmod10/sarma}
A.~D. Sarma, A.~Jain, and D.~Srivastava.
\newblock {I4E}: Interactive investigation of iterative information extraction.
\newblock In {\em SIGMOD}, 2010.

\bibitem{vldb07/shen}
W.~Shen, A.~Doan, J.~Naughton, and R.~Ramakrishnan.
\newblock Declarative information extraction using {D}atalog with embedded
  extraction predicates.
\newblock 2007.

\bibitem{sigmod08/shen}
W.~Shen, A.~Doan, J.~Naughton, and R.~Ramakrishnan.
\newblock Towards best-effort information extraction.
\newblock In {\em SIGMOD}, 2008.

\bibitem{provenance}
W.-C. Tan.
\newblock {Provenance in Databases: Past, Current, and Future}.
\newblock {\em IEEE Data Engineering Bulletin}, 2008.

\end{thebibliography}
